\documentclass[onecolumn, a4paper, accepted=2026-08-27]{quantumarticle}
\pdfoutput=1

\usepackage[utf8]{inputenc}
\usepackage[english]{babel}
\usepackage[T1]{fontenc}
\usepackage{hyperref}
\usepackage{mathtools}
\usepackage{amsfonts}
\usepackage{amsthm}
\usepackage{braket}
\usepackage{multirow}
\usepackage[ruled, vlined, noend]{algorithm2e}
\usepackage{subcaption}
\usepackage{adjustbox}
\usepackage{tikz}
\usepackage{makecell}

\usepackage[compat=newest]{yquant}
\useyquantlanguage{groups}
\usetikzlibrary{quotes, fit}

\PassOptionsToPackage{compress}{natbib}
\usepackage[numbers]{natbib}

\newtheorem*{lemma}{Lemma}

\title{Sparse quantum state preparation with improved Toffoli cost}

\author{Felix Rupprecht}
\orcid{0009-0004-2738-9711}
\email{felix.rupprecht@dlr.de}
\affiliation{Institute of Quantum Technologies, German Aerospace Center, 89081 Ulm, Germany}
\author{Sabine Wölk}
\orcid{0000-0001-9137-4814}
\affiliation{Institute of Quantum Technologies, German Aerospace Center, 89081 Ulm, Germany}
\affiliation{Center for Integrated Quantum Science and Technology (IQST), Ulm University, 89081 Ulm, Germany}
\date{}

\begin{document}
\begin{abstract}
    The preparation of quantum states is one of the most fundamental tasks in quantum computing, and a key primitive in many quantum algorithms.
    Of particular interest to areas such as quantum simulation and linear-system solvers are sparse quantum states, which contain only a small number~$s$ of non-zero computational basis states compared to a generic state.
    In this work, we present an approach that prepares $s$-sparse states on $n$ qubits, reducing the number of Toffoli gates required compared to prior art.
    We work in the established framework of first preparing a dense state on a $\lceil{\log(s)}\rceil$-qubit sub‑register, and then mapping this state to the target state via an isometry, with the latter step dominating the cost of the full algorithm.
    The speed\mbox{-}up is achieved by designing an efficient algorithm for finding and implementing the isometry.
    The worst-case Toffoli cost of our isometry circuit, which may be viewed as a batched version of an approach by Malvetti et al., is essentially $2s$ for sufficiently large values of $n$, yielding roughly a $\log(s)/2$ improvement factor over the state-of-the-art.
    In numerical benchmarks on randomly chosen states, the cost is closer to $s$.
    With the improved isometry circuit, we examine the dense-state preparation step and present ways to optimize the joint cost of both steps, particularly in the case of target states with purely real coefficients, by outsourcing some sub-tasks from the dense-state preparation to the isometry.
\end{abstract}
\maketitle

\section{Introduction}
\label{sec:introduction}

Preparing a quantum state~$\ket{\Theta}$, i.e., finding a quantum circuit $G$ such that $\ket{\Theta} = G \ket{0}$, is a key primitive in many areas of quantum computing, with applications ranging from linear-system solvers~\cite{morales2025} and quantum machine learning~\cite{Schuld2015} to quantum simulation~\cite{berry2024, Fomichev2024}.
For this reason, different approaches for preparing generic quantum states~\cite{Low2024, sun2023, zhang2022, Gui2024} and states in particular formats, such as matrix-product-states~\cite{berry2024, Fomichev2024, Schoen2005, Malz2024, Smith2024}, have been developed.
While many works~\cite{sun2023, zhang2022, Gui2024, Malz2024, Smith2024, Mao2024, Gleining2021, Ramacciotti2024, Mozafari2022, li2025} concentrate on optimizing the number of one- and two-qubit gates and/or the depth of their state-preparation circuits, the primary figure of merit for fault-tolerant quantum computation is often taken to be the number of Toffoli and/or T-gates~\cite{berry2024, Fomichev2024}.
This is due to the fact that those gates are usually more costly to implement in a fault-tolerant manner compared to Clifford gates, as they are created from $\ket{\text{CCZ}}$, respectively $\ket{\text{T}}$ resource states built with protocols like magic-state distillation~\cite{Litinski2019} and cultivation~\cite{gidney2024}.
We follow this approach and decompose all circuits into the Clifford+Toffoli gate set, taking the number of qubits required, together with the number of Toffoli gates needed, as our figures of merit.
As two $\ket{\text{T}}$ states may be catalytically created from a $\ket{\text{CCZ}}$ state and conversely a $\ket{\text{CCZ}}$ state from four $\ket{\text{T}}$ states~\cite{Beverland2020}, the Toffoli and T gate counts can be easily converted from one to another.

A class of quantum states of particular interest~\cite{morales2025, Fomichev2024} is the class of sparse quantum states.
An $s$-sparse quantum state on $n$~qubits is a linear combination $\ket{\Theta} = \sum_{i=0}^{s-1} c_i \ket{C_i}$ of computational basis states with $s \ll 2^n$, i.e., the number of computational basis states with non-zero amplitude is small compared to a generic state in the Hilbert space.
Approaches for efficiently preparing such sparse states have been studied in multiple works~\cite{Fomichev2024, Mao2024, Gleining2021, Ramacciotti2024, Mozafari2022, li2025, Malvetti2021, Tubman2018, deVeras2022} and can, for the most part, be categorized into two groups.
In~\cite{Mao2024, Gleining2021, Tubman2018, deVeras2022}, the sparse state is built iteratively in $s$~steps, where each step prepares a single basis state with the corresponding coefficient.
The non-Clifford gate cost of each iteration step is mostly given by the implementation of a multi-controlled rotation gate or an equivalent construction.
As detailed in Appendix~\ref{sec:other_papers}, the worst-case number~$w$ of control qubits for the multi-controlled gates is given by $n$ in \cite{Tubman2018}, $\lceil n/\log(n)\rceil + 1$ in~\cite{Mao2024}, $\lceil{\log(s)}\rceil$ in~\cite{Gleining2021} and the Hamming weight of the respective state in~\cite{deVeras2022}.
Here and in the rest of this work the logarithm is base-$2$.
The $w$-controlled rotations can be implemented via a ladder of $w-1$ AND gates~\cite{Babbush2018, Gidney2018} and a controlled rotation. 
If we allow an approximation error of $\epsilon$ for the whole preparation circuit, each controlled rotation may be implemented with Toffoli cost $\lceil \log(s/\epsilon) \rceil -2$ within the well-known~\cite{Low2024, Sanders2020} phase gradient framework (cf. Section~\ref{sec:further_optimizations}) or alternatively via synthesized approximation~\cite{ross2016, Kliuchnikov2023} of two uncontrolled rotations at similar cost.
The overall Toffoli count of this group of approaches hence exceeds $s\lceil \log(s)\rceil$.
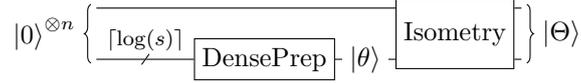
\begin{figure}[t]
    \centering
    \begin{tikzpicture}
        \begin{yquant}
            qubit {} rest;
            qubit {} enum;
            init {$\ket0^{\otimes n}$} (enum, rest);
            ["north:$\lceil{\log(s)}\rceil$" {font=\protect\footnotesize, inner sep=1pt}]
            slash enum;
            box {DensePrep} enum;
            text {$\ket{\theta}$} enum;
            box {Isometry} (enum, rest);
            output {$\ket{\Theta}$} (enum, rest);
        \end{yquant}
    \end{tikzpicture}
    \caption{Sparse state-preparation circuit. 
    First, the coefficients of the sparse state~$\ket{\Theta}$ are encoded into a dense state $\ket{\theta}$ within a subspace register.
    Subsequently, this dense state is transformed into the target state~$\ket{\Theta}$ via an isometry.}
    \label{fig:preparation_circuit}
\end{figure}

\begin{figure}[b]
    \centering
    \includegraphics{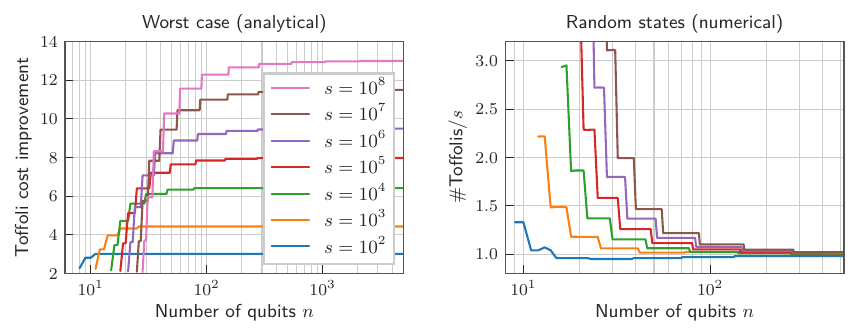}
    \caption{(Left) Improvement factor of the Toffoli cost for implementing the isometry step in this work \eqref{eq:toffoli_cost} over the method of~\cite{Malvetti2021}. 
    For both approaches, we considered the worst-case cost. 
    Since the upper bound \eqref{eq:toffoli_cost} is quite loose when $s$ and $m$ are close, the improvement according to \eqref{eq:toffoli_cost} reaches a maximum at a certain $n$ and then decreases, even though the actual improvement grows monotonically with $n$ toward an asymptote. 
    At each $n$, we therefore take the maximum improvement factor from all $n' \leq n$ in order to avoid this misrepresentation. 
    (Right) Number of Toffolis required for implementing the isometry step with our algorithm for random $s$-sparse states on $n$~qubits divided by $s$.}
    \label{fig:isometry_improvement}
\end{figure}

Instead of implementing $s$ successive rotations, the second group of algorithms~\cite{Fomichev2024, Ramacciotti2024, li2025, Malvetti2021}, which includes the approach presented in this paper, consists of a two-step process depicted in Figure~\ref{fig:preparation_circuit}.
In \emph{Step~1}, a dense state $\ket{\theta} = \sum_{i=0}^{s-1} c_{i} \ket{f(i)}$ is prepared in a $\lceil{\log(s)}\rceil$-qubit subspace register, with $f$ being a bijection from $[s]$ to an $s$-element subset of $[2^{\lceil{\log(s)}\rceil}]$.
Note that by $[k]$, we denote the set of integers $\{0, \dots, k-1\}$.
Then, in the subsequent \emph{Step~2}, an isometry is applied, mapping $\ket{f(i)}$ to $\ket{C_{i}}$, yielding the desired state $\ket{\Theta}$.
The Toffoli cost of this approach is the cost of the dense state preparation, which is usually implemented with a Grover-Rudolph style algorithm~\cite{Low2024}, plus the cost of the isometry.
The costs for the different isometry approaches in the references are investigated in Appendix~\ref{sec:other_papers} and listed in Table~\ref{tab:costing}.
Since in~\cite{Fomichev2024, Ramacciotti2024} the bijection $f$ is chosen to be the identity, a generic permutation of basis states must be implemented as the isometry circuit in their approaches.
In~\cite{Fomichev2024}, this incurs a Toffoli cost for the isometry of $2s(\lceil{\log(s)}\rceil -1 ) + \tilde{s}$, where here and in the rest of the paper, we define $\tilde{s}\coloneqq 2^{\lceil \log(s)\rceil}$.
That said, the freedom of choosing a non-trivial bijection allows for reducing the cost of the isometry step.
Indeed, as we will recall in Section~\ref{sec:isometry_circuit}, the method of~\cite{Malvetti2021} with the multi-controlled X gates implemented via AND ladders, achieves a worst-case isometry cost of $s(\lceil{\log(s)}\rceil - 1)$ Toffolis while requiring $\lceil{\log(s)}\rceil - 2$ ancilla qubits.
However, with the dense state preparation in Step~1 being implementable with $2.5\tilde{s}$ Toffolis or better (cf. Section~\ref{sec:further_optimizations}), even in this improved variant the isometry is still the bottleneck of the full algorithm and warrants further optimization.
Since~\cite{Malvetti2021} gives the lowest Toffoli count of the second group of algorithms and the cost of the approaches in the first group exceeds $s\lceil \log(s)\rceil$ we take reference~\cite{Malvetti2021} as the baseline for the rest of the paper.

We note that there are other approaches to sparse state preparation, potentially with a smaller number of required gates~\cite{Fomichev2024, zhang2022, vilmart2025resourceefficientsynthesissparsequantum}.
However, the proposals require a substantial amount of additional qubits.

The contributions of this paper are as follows.
We improve the Toffoli cost of sparse state preparation by constructing isometry circuits with the worst-case Toffoli count bounded by
\begin{equation}
\label{eq:toffoli_cost}
\biggl \lceil \frac{s}{m} \biggr \rceil \left(2m +\frac{\log(\tilde{s}/m)}{2}-3\right) + \log \left(\frac{\sqrt{\tilde{s}^3}}{2\sqrt{m}}\right)
\end{equation}
and ancilla cost $\lceil{\log(s)}\rceil - 1$.
Here, $m\coloneqq \max_{p\in \mathbb{N}}(2^p \mid 2^p \leq n - \lceil\log(s)\rceil)$ denotes the size of the largest register that fits in the $n$ system qubits without the $\lceil\log(s)\rceil$-qubit subspace register and whose qubit number is a power of two.
The Toffoli-cost improvement for the isometry circuit compared to~\cite{Malvetti2021} is plotted in Figure~\ref{fig:isometry_improvement} on the left.
The algorithm enabling these savings may be seen as a batched version of the approach in~\cite{Malvetti2021}.
There, in the inverse circuit of the isometry, sequentially, each target basis state in $\ket{\Theta}$ is mapped from the full space into the subspace, requiring a $\lceil{\log(s)}\rceil$-controlled-X gate for setting the qubits outside the subspace register to zero.
The main idea of this paper is to create batches of states of size $m$ and, instead of using $m$ individual multi-controlled-$X$ gates for the zeroing, use a single call of a partial unary iteration circuit for doing so.
The classical algorithm for finding the isometry circuit with classical runtime $O(s^2n)$ is explained in Section~\ref{sec:isometry_circuit}.
Applying the algorithm to random $s$-sparse states in Figure~\ref{fig:isometry_improvement} (right) shows Toffoli counts approaching $s$ for sufficiently large $n$, which is roughly half the worst-case cost.
\begin{table}[b]
    \setcellgapes{2pt}
    \makegapedcells
    \centering
    \resizebox{\textwidth}{!}{%
    \begin{tabular}{|c||c|c|c||c|c|}
        \hline
        \multirow{2}{*}{State} & \multicolumn{3}{c||}{Isometry} & \multicolumn{2}{c|}{Dense state preparation}\\
        \cline{2-6}
        & Toffolis & Qubits & Reference & Toffolis & Qubits \\
        \hline
        \hline
        \multirow{9}{*}{generic} & $2s(\lceil\log(s)\rceil - 1) + \tilde{s}$ & $n + 5\lceil\log(s)\rceil-3$ &~\cite{Fomichev2024} &  \multirow{5}{*}{$2\tilde{s}/2^{r} + b(\lceil\log(s)\rceil - 1) (2^r -1)$} & \multirow{11}{39pt}{$b2^r -r\; +$\\$2\lceil\log(s)\rceil$}\\
        \cline{2-4}
        & $s(\lceil\log(s)\rceil-1)$ & $n + \lceil\log(s)\rceil-2$ &~\cite{Malvetti2021} & &\\
        \cline{2-4}
        & $\lceil \frac{s}{m}\rceil(2m +\log(\tilde{s}/m)-3)$ & $n + \lceil\log(s)\rceil + 1$ & \makecell{rPUI +\\sign fix} & &\\
        \cline{2-5}
        & $3s(n-1)$ & $2n-1$ &~\cite{Ramacciotti2024} * & \multirow{5}{*}{$3\tilde{s}/2^{r} + (b\lceil\log(s)\rceil -b +1) (2^r -1)$} &\\
        \cline{2-4}
        & $2s(\lceil\log(2s)\rceil-1) +n - \lfloor\log(2s)\rfloor$ & $2n-\lfloor\log(2s)\rfloor$ &~\cite{li2025} * & &\\
        \cline{2-4}
        & \makecell{$\lceil \frac{s}{m}\rceil(2m +\frac{1}{2}\log(\tilde{s}/m)-3)$ \\ $+ \frac{1}{2}\log(\tilde{s}^3/4m)$} & $n + \lceil\log(s)\rceil - 1$ & uPUI & &\\
        \cline{1-5}
        real & $\lceil \frac{s}{m}\rceil(2m +\log(\tilde{s}/m)-3)$ & $n + \lceil\log(s)\rceil + 1$ & \makecell{rPUI +\\phase enc.}& $\tilde{s}/2^{r} + b(\lceil\log(s)\rceil - 2) (2^r -1)$ &\\
        \hline
    \end{tabular}%
    }
    \caption{Upper bounds for the Toffoli and qubit costs of the isometry and dense-state-preparation steps for generic states and states with real coefficients. 
    The rows without a reference correspond to our algorithm variants explained in Section~\ref{sec:further_optimizations}: restricted PUI with integrated QROAM sign-fix, unrestricted PUI and lastly restricted PUI with integrated QROAM sign-fix and phase encoding. 
    The counts for the entries marked with~* are derived in Appendix~\ref{sec:other_papers}. We write $\tilde{s}\coloneqq 2^{\lceil \log(s)\rceil}$ and let $m$ be the greatest power of two such that $m \leq n - \lceil\log(s)\rceil$. 
    For the first two isometry approaches from the literature, the leading Toffoli-cost term of the dense state preparation may be assumed to be $2$ instead of $3$, since a slight modification of the isometry circuits, analogous to our work in Section~\ref{sec:further_optimizations}, allows the dense state preparation registers to be uncomputed at no extra cost.
    For simplicity, we choose a single~$r$ for all clean QROAM calls in the dense state preparation, where $b$ is the number of bits for the rotation angles. We exclude the costs for preparing the phase-gradient state and do not count its qubits. 
    Moreover, we assume that the number of qubits is largest during the QRO(A)Ms and not during the rotations.}
    \label{tab:costing}
\end{table}

A partial unary iteration circuit can be seen as a variant of the standard unary iteration circuit $U = \sum_{i=0}^{2^u-1} \ket{i}\bra{i} \otimes U_i$ proposed in~\cite{Babbush2018}.
It does not iterate over all basis states $[\ket{0},\ket{2^u}]$ of a $u$-qubit address register, but only over an interval $[\ket{l},\ket{r}]$.
Throughout the paper, the square brackets denote intervals of consecutive basis states ordered according to their binary representation in big-endian.
In Section~\ref{sec:partial_unary_iteration}, we consider two types of such partial unary iteration (PUI) circuits, which we call restricted and unrestricted, and analyze their cost.
A restricted partial unary iteration circuit guarantees that $U_i$ is non-trivial only for states $\ket{i} \in [\ket{l},\ket{r}]$ within the target interval, while the unrestricted version, with an improved Toffoli count, relaxes this guarantee by allowing the circuit to have non-trivial contributions $\ket{i}\bra{i}\otimes \tilde{U}_i$ for $i>r$.
If the unrestricted version is used, one needs to keep track of those contributions so that they can be taken into account when building subsequent parts of the circuit.
For the isometry circuit in Section~\ref{sec:isometry_circuit} with the Toffoli bound~\eqref{eq:toffoli_cost}, unrestricted PUIs were applied.

With the reduced cost of the isometry, the dense state preparation can now be a major contributor to the total cost of the full sparse state preparation circuit.
Therefore, in Section~\ref{sec:further_optimizations}, we review recent dense state preparation techniques with their Toffoli/qubit trade-offs and analyze ways of reducing the overall cost of the full sparse state preparation algorithm.
By replacing the unrestricted partial iterations in the isometry with restricted ones, we can move parts of the dense state preparation to the isometry circuit.
This is particularly useful in the case of real states, where it allows us to insert the complete phase information of the state in the isometry, reducing the Toffoli cost of the dense state preparation step by a factor of up to three.
The resulting worst case resource counts for the algorithms are given in Table~\ref{tab:costing}. 

Quantum circuits for the sparse state preparation and the necessary sub-circuits are implemented within the \texttt{Qualtran} framework~\cite{harrigan2024}, with the computationally more expensive classical algorithms for finding the isometry outsourced to SIMD-accelerated Rust code.
The full code and other assets created for this paper are available on Zenodo~\cite{rupprecht2025}.

\section{Isometry circuit}
\label{sec:isometry_circuit}

In this section, we present the classical algorithm for constructing the improved isometry circuit, together with the bijection $f$.
We start with the classically given target state $\ket{\Theta} = \sum_{i=0}^{s-1} c_i \ket{C_i}$ and think of the computational basis states $\ket{C_i}$ as binary bitstrings in a tableau consisting of $s$~rows and $n$ columns, with the state $\ket{C_i}$ constituting row $i$.
We start indexing at $0$ from the left and denote by $\ket{e_k}$ the computational basis state with qubit $k$ set to $\ket{1}$ and all other qubits set to $\ket{0}$.

We divide the bitstrings $\ket{C_i} = \ket{C_i'}^s\ket{C_i''}^r$ at position $\lceil{\log(s)}\rceil$, such that $\ket{C_i'}^s$ is a state within the $\lceil{\log(s)}\rceil$-qubit subspace register and $\ket{C_i''}^r$ a state on the remaining qubits that make up the non-subspace register.
The task of the algorithm is now to find a circuit and a bijection $f$ such that each state $\ket{C_i} = \ket{C_i'}^s\ket{C_i''}^r$ gets mapped to $\ket{f(i)}^s\ket{0}^r$, i.e., we want to map the target state into the subspace.
The isometry circuit in Step 2 is then simply the inverse of the circuit found by the algorithm.
\begin{figure}[b]
    \centering
    \begin{tikzpicture}
        \begin{yquant*}[every control/.append style={radius=0.8mm}]
            qubit {} reg[7];
            init {$\begin{matrix}
            \ket{011}^s\ket{1011}^r\\
            \ket{111}^s\ket{0111}^r\\
            \ket{100}^s\ket{1010}^r\\
            \ket{011}^s\ket{0110}^r\\
            \ket{101}^s\ket{1101}^r\\
            \ket{100}^s\ket{1111}^r\\
            \ket{111}^s\ket{1010}^r
            \end{matrix}\eqcolon \ket{\Theta}$} (-);
            cnot reg[5-6] | reg[3];
            cnot reg[3] | reg[1-2] ~ reg[0];
            align -;
            cnot reg[5-6] | reg[4];
            cnot reg[4] | reg[0-2];
            align -;
            cnot reg[6] | reg[3];
            cnot reg[3] | reg[0] ~ reg[1-2];
            align -;
            cnot reg[4], reg[2] | reg[3];
            cnot reg[3] | reg[1] ~ reg[2], reg[0];
            align -;
            cnot reg[2] | reg[3];
            cnot reg[3] |reg[2], reg[0] ~ reg[1];
            align -;
            cnot reg[1], reg[5] | reg[4];
            cnot reg[4] | reg[0-1] ~ reg[2];
            align -;
            cnot reg[0-2] | reg[3];
            cnot reg[3] ~ reg[0-2];
            align -;
            output {$\;\begin{matrix}
            \ket{011}^s\ket{0000}^r\\
            \ket{111}^s\ket{0000}^r\\
            \ket{100}^s\ket{0000}^r\\
            \ket{010}^s\ket{0000}^r\\
            \ket{101}^s\ket{0000}^r\\
            \ket{110}^s\ket{0000}^r\\
            \ket{000}^s\ket{0000}^r
        \end{matrix}$} (-);
        \end{yquant*}
    \end{tikzpicture}
    \caption{Isometry circuit created by the algorithm in~\cite{Malvetti2021}, mapping $\ket{\Theta}$ to a $3$-qubit subspace. 
    For each element outside the subspace, apply CX gates to obtain $\ket{k}^s\ket{e_j}^r$ for a $\ket{k}^s$ not yet present in the subspace and a $j\in \{0, \dots, 3\}$. Then, use a multi-controlled-X gate to transform $\ket{k}^s\ket{e_j}^r$ to $\ket{k}^s\ket{0}^r$. 
    We index qubits from top to bottom and bitstrings from left to right.}
    \label{fig:malvetti}
\end{figure}
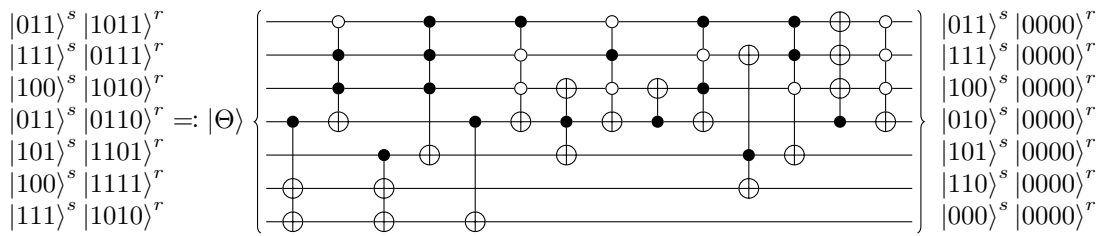

Before giving our algorithm for finding the improved isometry circuit, we briefly recall the approach of Malvetti et al.~\cite{Malvetti2021}, with an example circuit shown in Figure~\ref{fig:malvetti}.
They iterate over all states~$\ket{C_i} = \ket{C_i'}^s\ket{C_i''}^r$ and apply the following logic.
If $\ket{C_i''}^r$ is the zero state, i.e., the state~$\ket{C_i}$ already lies in the subspace, set $f(i)\coloneqq C_i'$.
Otherwise, it has a non-zero qubit at position $l$, and, using CX gates controlled on this qubit, one can transform $\ket{C_i} = \ket{C_i'}^s\ket{C_i''}^r$ into $\ket{k}^s\ket{e_l}^r$ for some basis state $\ket{k}^s$ that is not yet present in the subspace register.
Note that this does not affect states already in the subspace, as the control qubit in the non-subspace register is $\ket{0}$ for those states.
The iteration step then finishes by setting $f(i)\coloneqq k$ and zeroing the lone $\ket{1}$-qubit in the non-subspace register of $\ket{k}^s\ket{e_l}^r$ with a $\lceil\log(s)\rceil$-controlled-X gate controlled on the content~$\ket{k}$ of the subspace register.
This yields $\ket{k}^s\ket{0}^r$ as desired, while leaving states already in the subspace untouched by the choice of $\ket{k}^s$.
The worst-case Toffoli cost of this approach is given by $s(\lceil\log(s)\rceil - 1)$ if the multi-controlled-X gates are decomposed into Toffoli gates with the AND-gate construction of~\cite{Babbush2018, Gidney2018}, requiring $\lceil\log(s)\rceil - 2$ ancilla qubits.
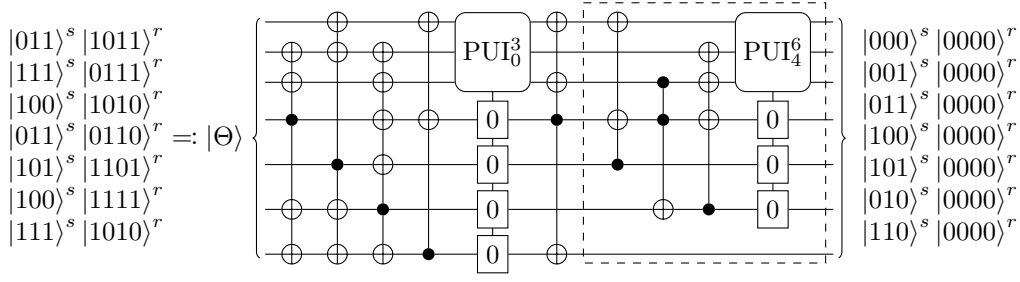
\begin{figure}[b]
    \centering
    \begin{tikzpicture}
        \begin{yquant*}[every control/.append style={radius=0.8mm}]
            qubit {} reg[7];
            init {$\begin{matrix}
            \ket{011}^s\ket{1011}^r\\
            \ket{111}^s\ket{0111}^r\\
            \ket{100}^s\ket{1010}^r\\
            \ket{011}^s\ket{0110}^r\\
            \ket{101}^s\ket{1101}^r\\
            \ket{100}^s\ket{1111}^r\\
            \ket{111}^s\ket{1010}^r
            \end{matrix}\eqcolon \ket{\Theta}$} (-);
            cnot reg[1-2], reg[5-6] | reg[3];
            cnot reg[5-6], reg[0-1] | reg[4];
            cnot reg[6], reg[1-4] | reg[5];
            cnot reg[0], reg[3] | reg[6];
            [name=one, operator style={only at={0}{rounded corners}}] box {\Ifnum\idx<1 PUI$_0^3$\Else $0$\Fi} (reg[0-2]), reg[6], reg[5], reg[4], reg[3];
            cnot reg[0], reg[6], reg[2] | reg[3];
            [this subcircuit box style={draw, dashed}]
                subcircuit {
                qubit {} reg[7];
                cnot reg[0], reg[3] | reg[4];
                cnot reg[5] | reg[3], reg[2];
                cnot reg[1-3]| reg[5];
                [name=two, operator style={only at={0}{rounded corners}}] box {\Ifnum\idx<1 PUI$_4^6$\Else $0$\Fi} (reg[0-2]), reg[5], reg[4], reg[3];
                \draw (two-0) -- (two-3) -- (two-2) -- (two-1);
                } (-);
            output {$\;\begin{matrix}
            \ket{000}^s\ket{0000}^r\\
            \ket{001}^s\ket{0000}^r\\
            \ket{011}^s\ket{0000}^r\\
            \ket{100}^s\ket{0000}^r\\
            \ket{101}^s\ket{0000}^r\\
            \ket{010}^s\ket{0000}^r\\
            \ket{110}^s\ket{0000}^r
        \end{matrix}$} (-);
        \end{yquant*}
        \draw (one-0) -- (one-4) -- (one-3) -- (one-2) -- (one-1);
    \end{tikzpicture}
    \caption{Our isometry circuit mapping $\ket{\Theta}$ to a $3$-qubit subspace. 
    First, we create a batch $(\ket{0}^s\ket{e_0}^r, \dots, \ket{3}^s\ket{e_3}^r)$ using multi-controlled-CX gates and apply an unrestricted partial unary iteration PUI$_0^3$ with control interval $[\ket{0}^s, \ket{3}^s]$ to zero the qubits of the batch-states in the non-subspace register. 
    Throughout the paper, we use rounded corners on boxes to emphasize that the qubits entering and exiting the box are not changed but merely serve as control qubits for the other parts of the operation. 
    We repeat the procedure for the second batch $(\ket{4}^s\ket{e_0}^r, \dots, \ket{6}^s\ket{e_2}^r)$. The circuits for the PUIs are given in Figure~\ref{fig:unrestricted_pui}. 
    Moreover, the creation of the second batch within the dotted box is analyzed in detail in Figure~\ref{fig:batch}.}
    \label{fig:isometry}
\end{figure}

The main idea of this paper is to find a version of the previous algorithm that does not require $s$ individual calls to a multi-controlled-X gate performing the zeroing $\ket{k}^s\ket{e_l}^r \mapsto \ket{k}^s\ket{0}^r$.
Instead, we sequentially create $\lceil s/m \rceil$ length-$m$ batches of the form $(\ket{k}^s\ket{e_0}^r, \dots, \ket{k+m-1}^s\ket{e_{m-1}}^r)$ and, for each batch, use a partial unary iteration circuit introduced in Section~\ref{sec:partial_unary_iteration} for realizing the zeroing operation
\begin{equation}
\label{eq:zeroing}
    (\ket{k}^s\ket{e_0}^r, \dots, \ket{k+m-1}^s\ket{e_{m-1}}^r) \mapsto (\ket{k}^s\ket{0}^r, \dots, \ket{k+m-1}^s\ket{0}^r).
\end{equation}
We will see in Section~\ref{sec:partial_unary_iteration} that this partial unary iteration is most efficient if $m$ is a power of two.
Hence we take it to be the largest power of two that is less than or equal to the size of the non-subspace register, i.e., $m \leq n - \lceil\log(s)\rceil$.
A formal description of the classical algorithm for finding the isometry circuit and bijection is given in Algorithm~\ref{alg:isometry_finding}, with an example circuit shown in Figure~\ref{fig:isometry}.

If there are states not yet in the subspace, we start building a new batch by finding a state with the first qubit of the non-subspace register set to $\ket{1}$, potentially swapping qubits to do so.
Using a multi-target-CX gate, this state is transformed into $\ket{k}^s\ket{e_0}^r$, where $k$ is a runner variable that is initialized to zero at the beginning of the algorithm and incremented with each element added to a batch.
The resulting state then constitutes the first element of the batch.
In this way, the foremost multi-target CX in Figure~\ref{fig:isometry} transforms the state $\ket{110}^s\ket{1011}^r$ into $\ket{000}^s\ket{1000}^r$.
If there are still states not in the subspace or the batch, and the batch is not full, we take a state with the second qubit of the non-subspace register set to $\ket{1}$, assuming for now that such a state exists, and apply a multi-target CX controlled on this $\ket{1}$ to yield the second batch element $\ket{k+1}^s\ket{e_1}^r$.
Repeating the process for the third up to the $m$-th qubit of the non-subspace register, or until all states are either in the batch or in the subspace, we end up with the batch $(\ket{k}^s\ket{e_0}^r, \dots, \ket{k+m-1}^s\ket{e_{m-1}}^r)$.
Then we use a partial unary iteration circuit PUI$_k^{k+m-1}$ over the address interval $[\ket{k}, \ket{k+m-1}]$ to implement the mapping \eqref{eq:zeroing}, zeroing out all qubits of the batch-states outside the subspace register.
If we do not fill up the batch completely, $m$ needs to be replaced by the number of elements in the batch.
The circuits for the partial unary iteration are given in Section~\ref{sec:partial_unary_iteration}.
They leave all states~$\ket{j}^s\ket{0}^r$ for $j < k$ unchanged.
However if the unrestricted version of PUI is used, states~$\ket{j}^s\ket{0}^r$ with $j \geq k+m$ may be altered, and the tableau must keep track of those changes.
In the algorithm, either the restricted or the unrestricted version of PUI may be used.
However, unless stated otherwise, we use unrestricted PUIs because of their improved Toffoli cost.
After the zeroing is done, the batch is finished, and we start building a new one.

\begin{figure}[t]
    \centering
    \begin{adjustbox}{width=\textwidth}
        $\begin{matrix}
                \ket{100}^s\ket{1000}^r\\
                \ket{001}^s\ket{1100}^r\\
                \ket{101}^s\ket{1000}^r\\
            \end{matrix}
            \xrightarrow[]{\text{CMX}_4^{0,3}}
            \begin{matrix}
                \textcolor{purple}{\ket{100}^s\ket{1000}^r}\\
                \ket{101}^s\ket{0100}^r\\
                \textcolor{purple}{\ket{101}^s\ket{1000}^r}\\
            \end{matrix}
            \xrightarrow[]{\text{Tof}_{2,3}^{5}}
            \begin{matrix}
                \ket{100}^s\ket{1000}^r\\
                \ket{101}^s\ket{0100}^r\\
                \ket{101}^s\ket{1010}^r\\
            \end{matrix}
            \xrightarrow[]{\text{CMX}_5^{1,2,3}}
            \begin{matrix}
                \ket{100}^s\ket{1000}^r\\
                \ket{101}^s\ket{0100}^r\\
                \ket{110}^s\ket{0010}^r\\
            \end{matrix}
            \xrightarrow[]{\text{PUI}_{4}^{6}}
            \begin{matrix}
                \ket{100}^s\ket{0000}^r\\
                \ket{101}^s\ket{0000}^r\\
                \ket{110}^s\ket{0000}^r\\
            \end{matrix}$
    \end{adjustbox}
    \caption{Creation of the second batch of the example in Figure~\ref{fig:isometry} (dotted box). 
    The first state is already in the correct form $\ket{4}^s\ket{e_0}^r$.
    For the second, we use a multi-target-CX gate (CMX) to bring it into batch form. 
    Hereby, the upper indices denote target qubits, while the lower indices denote the control qubits. 
    The last state requires a Toffoli gate controlled on the third qubit, where the purple states differ, and on the first subspace qubit, where they coincide. 
    Subsequently, the last state is transformed so that the states constitute a full batch. 
    The non-subspace register of the batch elements is then set to zero via an unrestricted partial unary iteration.}
    \label{fig:batch}
\end{figure}

Looking at the example of Figure~\ref{fig:isometry}, four of the basis states of $\ket{\Theta}$, constituting the first batch, are mapped into the subspace and there remain three states still to be mapped.
In Figure~\ref{fig:batch} we give a detailed description of the creation of the new batch, starting with the process of adding the second element.
Above, we assumed that with $(l-1)$ elements in the batch there is a state with the $l$-th qubit of the non-subspace register set to $\ket{1}$.
As can be seen in the example, this is not always the case, and we need to create such a state without altering the ones in the subspace and the batch.
If there is a state with a $\ket{1}$-qubit in the non-subspace register further right than the $l$-th qubit, a simple swap suffices.
Otherwise, we pick any $\ket{C}=\ket{C'}^s\ket{C''}^r$ not in the batch or the subspace and note that $\ket{C''}^r$ has a $\ket{1}$ within the first $(l-1)$ qubits, e.g., qubit $j$.
By construction, there is a state $\ket{r}^s\ket{e_j}^r$ in the batch that must differ from $\ket{C}$ at some qubit $u$.
In Figure~\ref{fig:batch} this state and $\ket{C}$ are marked in purple.
If $\ket{C}$ has a $\ket{1}$ at qubit $u$ we apply a Toffoli gate with controls $u$ and $j$ and the $l$-th non-subspace qubit as the target, creating the desired state with the $l$-th non-subspace qubit in $\ket{1}$.
If qubit $u$ of $\ket{C}$ is $\ket{0}$ instead, we let the Toffoli be negatively controlled on this qubit.

The bijection can be read off from the final tableau.
Zeroing the elements of the final batch requires some more care.
If there are no more elements outside the subspace, there are two cases.
If each state within the tableau has been in a batch, we can simply apply the zeroing operation to the current batch elements as usual.
However, if states lie in the subspace without having been part of a batch, we must make sure that emptying the current batch does not map those elements out of the subspace.
We do so by using a restricted partial unary iteration for the zeroing operation in this case.
If, during this operation, with the last element in the batch being $\ket{r}^s\ket{e_j}^r$, a non-batch state $\ket{r}^s\ket{0}^r$ is present, one ends up with the non-subspace state $\ket{r}^s\ket{e_j}^r$ after the zeroing.
To fix this and map the state into the subspace, we use a single step of the Malvetti et al. algorithm with Toffoli cost $\lceil\log(s)\rceil + 1 = \log(\tilde{s}/2)$.

The Toffoli cost of the resulting circuit is the cost of the partial unary iterations plus at most
\begin{equation}
\label{eq:algo}
s - \lceil s/m \rceil + \log(\tilde{s}/2)
\end{equation}
Toffolis.
Note that the negative term is a consequence of the fact that the first element in each batch is guaranteed not to require a Toffoli.
Moreover, only for the partial unary iteration circuits ancilla qubits are necessary.
In numerical examples on random states (cf. Figure~\ref{fig:isometry_improvement}), we observe that the actual cost bounded by \eqref{eq:algo} is considerably lower, because suitable non-zero columns exist most of the time.
When it comes to the multi-target CX gates in the circuit, it is worth mentioning that, using lattice surgery within a surface code, CX gates with multiple targets do not take longer than single-target CX gates~\cite{Litinski2018}.

For the classical cost of the algorithm, each basis state requires the tableau to be updated at most three times: once after the zeroing operation and at most twice while building the batch.
Each of those updates may be performed at a cost linear in the tableau size, i.e., in $O(sn)$, and can be parallelized.
Since, for all other operations, such as finding states not yet in the batch, it also suffices to scan the tableau once, and we have a constant number of such scans per basis state, the classical asymptotic cost of the algorithm is $O(s^2 n)$.
In Figure~\ref{fig:numerics_isometry} (right), we collect some wall-clock times for our implementation~\cite{rupprecht2025} of the algorithm in Rust using SIMD intrinsics.
On a CPU node with two AMD EPYC 7452 processors, the algorithm for $s=10^7$ non-zero basis states runs in about 10-20 minutes.
Making profitable use of multi-threading for the algorithm is non-trivial, because each individual parallelizable row operation is very cheap.
We expect there to be room for improvement in the run-time of the classical algorithm, potentially by using GPUs.

\section{Partial unary iteration}
\label{sec:partial_unary_iteration}

In this section, we construct the partial unary iteration circuits (PUI) used in Section~\ref{sec:isometry_circuit} and analyze their cost.
Given an interval of states $[\ket{l}^a, \ket{r}^a]$ within an $w$-qubit address register, the goal is to apply unitaries $U_l, \dots, U_r$ on a target register controlled on the values in the address register, i.e., we want to find a circuit $U$ implementing
\begin{equation}
\label{eq:mapping}
U (\ket{i}^a\otimes \ket{t}^t) = \ket{i}^a\otimes(U_i \ket{t}^t)
\end{equation}
for all $i\in \{l,\dots,r\}$.

If $l=0$ and $r=2^w - 1$, this task is solved by standard unary iteration circuits based on the traversal of balanced binary trees~\cite{Babbush2018} or, alternatively, skew trees~\cite{Khattar2025}.
Both approaches are nicely depicted in~\cite{Khattar2025}.
The following discussion is based on the original unary iteration circuits~\cite{Babbush2018} that utilize balanced binary trees.
The Toffoli cost of this approach is $2^w - 2$ (or $2^w - 1$ if controlled) and requires $w-1$ ancillas.
An example of such a full unary iteration circuit for $w=3$ is given in Figure~\ref{fig:restricted_pui}, where we chose $U_0, \dots, U_6$ to be X gates on different qubits of the target register and $U_7$ to be the identity.

\begin{figure}[b]
    \centering
    \begin{adjustbox}{width=\textwidth}
        \begin{tikzpicture}
            \begin{yquant}[every control/.append style={radius=0.8mm}]
                qubit {$a_{\The\numexpr\idx}$} enum[3];
                nobit k;
                nobit l;
                qubit {} rest[4];
                [operator/separation=-5.6mm]
                text {$d_0$} k;
                [operator/separation=-5.6mm]
                text {$d_1$} l;
                [every negative control/.append style={only at={0}{draw=purple}}]
                init k ~ enum[0-1];
                init l | k ~ enum[2];
                cnot rest[0] | l;
                cnot l | k;
                cnot rest[1] | l;
                discard l | k ~ enum[2];
                [every negative control/.append style={only at={0}{draw=purple}}]
                cnot k  ~ enum[0];
                init l | k ~ enum[2];
                cnot rest[2] | l;
                cnot l | k;
                cnot rest[3] | l;
                discard l | k, ~ enum[2];
                [every negative control/.append style={only at={0}{draw=purple}}]
                discard k ~ enum[0-1];
                barrier -;
                [name=zero]
                init k | enum[0] ~ enum[1];
                \enclose[to={l[0]}]{
                    align -;
                    settype {qubit} k;
                    [name=one]
                    init l | k ~ enum[2];
                    settype {qubit} l;
                    [name=two]
                    cnot rest[0] | l;
                    [name=three]
                    cnot l | k;
                    [name=four]
                    cnot rest[1] | l;
                    [name=five]
                    discard l | k ~ enum[2];
                    \coordinate (target_height) at (0, -4.1);
                    \node[font=\footnotesize] at (zero |- target_height) {(0)};
                    \node[font=\footnotesize] at (one |- target_height) {(1)};
                    \node[font=\footnotesize] at (two |- target_height) {(2)};
                    \node[font=\footnotesize] at (three |- target_height) {(3)};
                    \node[font=\footnotesize] at (four |- target_height) {(4)};
                    \node[font=\footnotesize] at (five |- target_height) {(5)};
                }
                cnot k | enum[0];
                \enclose[to={l[0]}, purple]{
                    settype {qubit} k;
                    [every negative control/.append style={only at={0}{draw=purple}}]
                    init l | k ~ enum[2];
                    settype {qubit} l;
                    cnot rest[2] | l;
                    cnot l | k;
                    [every negative control/.append style={only at={0}{draw=purple}}]
                    discard l | k ~ enum[2];
                }
                discard k | enum[0] ~ enum[1];
            \end{yquant}
        \end{tikzpicture}
    \end{adjustbox}
    \caption{Full unary iteration circuit~\cite{Babbush2018} applying X gates to different qubits of the lower 4-qubit target register, controlled on the values of the address register. 
    The left and right parts of the circuit can be seen as two restricted partial unary iteration circuits on the values $[\ket{0}^a, \ket{3}^a]$ and $[\ket{4}^a, \ket{6}^a]$, respectively. 
    In order to obtain the unrestricted versions of those PUIs shown in Figure~\ref{fig:unrestricted_pui}, one removes the purple control dots. 
    The purple dashed box is then simplified to the box shown in Figure~\ref{fig:unrestricted_pui}. 
    The definition of the AND gates is recalled in Appendix~\ref{sec:logical_and}.}
    \label{fig:restricted_pui}
\end{figure}
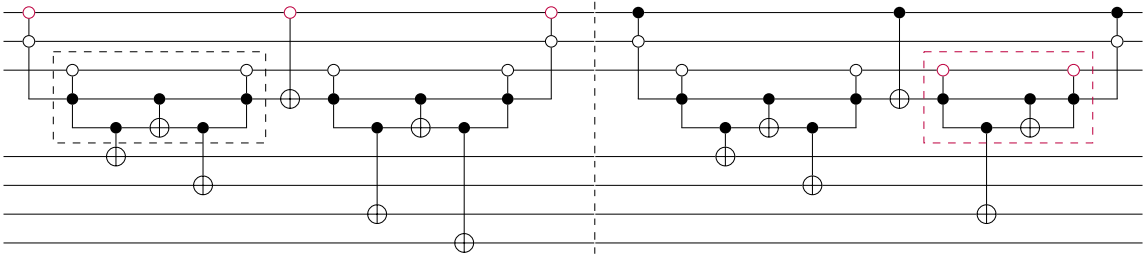

The core primitive of unary iteration, corresponding to the splitting of a node of the tree into its two children, may be seen in the dashed boxes within the figure.
For now, we concentrate on the left box with the black dashes.
Entering the box are the address qubit $a_2$ and an ancilla qubit~$d_0$ created earlier (Step 0) by the AND gate (cf. Appendix~\ref{sec:logical_and}) on the first two address qubits $a_0$ and~$a_1$.
The ancilla $d_0$ is in the state $\ket{1}^{d_0}$ if those two address qubits are in the state $\ket{1}^{a_0}\ket{0}^{a_1}$.
With our big-endian convention, this is the case if the $3$-qubit state in the address register lies within the interval $[\ket{4}^a, \ket{7}^a]$.
In general, the incoming ancilla qubit 'selects' an interval of address register values that corresponds to a node in the binary tree.
Within the box, by applying an AND gate (Step 1) controlled on the ancilla $d_0$ and negatively controlled on the incoming address qubit $a_2$, the ancilla qubit $d_1$ created by this new AND is $\ket{1}^{d_1}$ if the address value lies in the left half of the interval, i.e., in our case if the address register is in the state $\ket{4}^a$.
This left interval corresponds to one child of the node.
In our case, it is a leaf node, and we can apply $U_4$ (Step 2) controlled on the ancilla qubit $d_1$.
In the general case, if the half-interval at that point contains more than a single element, one recursively applies the procedure we just described until the interval has length one.
When doing so, one would take the current ancilla $d_1$ and the next address qubit, i.e., $a_3$ if it existed, as the incoming qubits of the next box.
In order for the ancilla $d_1$ in the current box to indicate that the address value lies in the right half-interval, i.e., in our case to be $\ket{5}^a$, it suffices (Step 3) to toggle it controlled on the prior ancilla $d_0$.
This half-interval corresponds to the other child node, and since it is a leaf node we can apply $U_5$ controlled on the ancilla (Step~4) and uncompute the AND gate (Step 5).
Note that the adjoint AND may be realized at zero Toffoli cost using measurement-based uncomputation.

If one now wants to restrict the iteration to a target interval of values $[\ket{l},\ket{r}]$, one can simply limit the traversal of the balanced binary tree to values within that interval, i.e., at each step we continue the recursion on the left/right half of the current interval only if the respective half-interval overlaps with the target interval.
Applying $U_l, \dots, U_{r}$ when reaching the corresponding leaf nodes then implements a unitary satisfying equation \eqref{eq:mapping}.
We choose to call the resulting circuit \emph{restricted partial unary iteration}.
Examples for such restricted PUIs are the left and right parts of the circuit in Figure~\ref{fig:restricted_pui}.

In Section~\ref{sec:isometry_circuit}, multiple partial unary iteration circuits over successive intervals are used, with the target operations being X gates applied to individual qubits of the target register.
In situations like this, it is possible to reduce the Toffoli count of the iterations by allowing the unitary
\begin{equation*}
    U = \sum_{i=l}^{r} \ket{i}\bra{i}^a\otimes U_i + \sum_{i\in L} \ket{i}\bra{i}^a\otimes U_{g(i)}
\end{equation*}
to have a non-trivial effect on address-register values that are greater than the upper limit $r$ of the interval we iterate over.
Here, $L \subseteq \{r+1, \dots, 2^w-1\}$ are the values of the address-register elements greater than $r$ for which $U$ does not act trivially on the target register.
The unitary $U_{g(i)}$ applied for those values is determined by the map $g\colon L \to \{l, \dots, r\}$.
In the following, we construct a circuit PUI$_l^r$ implementing such a unitary and call it \emph{unrestricted partial unary iteration}.

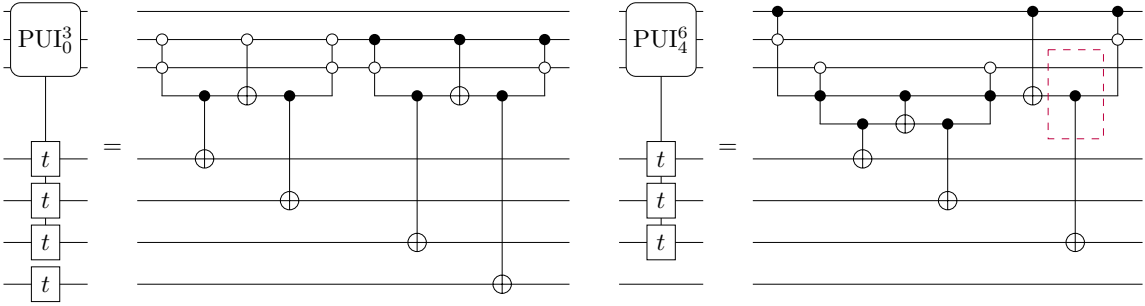
\begin{figure}[b]
    \centering
    \begin{adjustbox}{width=\textwidth}
        \begin{tikzpicture}         
            \begin{yquantgroup}[every control/.append style={radius=0.8mm}, drawing mode=size]
                \registers {
                    qubit {} enum[3];
                    nobit k;
                    nobit l;
                    qubit {} rest[4];
                }
                \circuit {
                    [name=a, operator style={only at={0}{rounded corners}}] box {\Ifnum\idx<1 PUI$_0^3$\Else $t$\Fi} (enum[0-2]), rest[0-3];
                    \draw (a-0) -- (a-1) -- (a-2) -- (a-3);
                }
                \equals
                \circuit {
                    init {$a_{\The\numexpr\idx}$} enum;
                    init k ~ enum[1-2];
                    cnot rest[0] | k;
                    cnot k | ~ enum[1];
                    cnot rest[1] | k;
                    discard k ~ enum[1-2];
                    init k | enum[1] ~ enum[2];
                    cnot rest[2] | k;
                    cnot k | enum[1];
                    cnot rest[3] | k;
                    discard k | enum[1] ~ enum[2];
                }
                \equals[$\;\;\;$]
                \circuit {
                    [name=b, operator style={only at={0}{rounded corners}}] box {\Ifnum\idx<1 PUI$_4^6$\Else $t$\Fi} (enum[0-2]), rest[0-2];
                    \draw (b-0) -- (b-1) -- (b-2) -- (b-3);
                }
                \equals
                \circuit {
                    init {$a_{\The\numexpr\idx}$} enum;
                    init k | enum[0] ~ enum[1];
                    \enclose[to={l[0]}]{
                        settype {qubit} k;
                        init l | k ~ enum[2];
                        settype {qubit} l;
                        cnot rest[0] | l;
                        cnot l | k;
                        cnot rest[1] | l;
                        discard l | k ~ enum[2];
                    }
                    cnot k | enum[0];
                    \enclose[purple]{
                        cnot rest[2] | k;
                    }
                    discard k | enum[0] ~ enum[1];
                }
            \end{yquantgroup}
        \end{tikzpicture}
    \end{adjustbox}
    \caption{The two unrestricted partial unary iteration circuits on the intervals $[\ket{0}^a, \ket{3}^a]$ and $[\ket{4}^a, \ket{6}^a]$ used in Figure~\ref{fig:isometry}. 
    They can be obtained from the restricted versions on the left and right-hand sides of Figure~\ref{fig:restricted_pui} via simplifications after removing the purple controls. 
    These removals do not change how the target-X unitaries are applied on address-register values in the current interval, leaving states with address values left of the interval untouched. 
    However, the first target‑X gate in PUI$_0^3$ is applied not only for the address value $\ket{0}^a$ but also for $\ket{4}^a$, i.e., a value outside the interval of the PUI, where the most‑significant bit changed from $0$ to $1$. 
    The same is true for all other target‑X gates in the first batch. }
    \label{fig:unrestricted_pui}
\end{figure}

Even though basis states $\ket{j}^a\ket{v}^t$ with $j>r$ can change, states with $j<l$ are left untouched.
The application of the unrestricted form of partial unary iteration makes sense in situations where one either knows that the state to which it is applied does not have overlap with address-register components greater than the interval, or when one iterates over the full address-register values from left to right and can undo unwanted effects on a value once it lies within an active interval.
For the latter, one must be able to efficiently keep track of these unwanted effects.
This is the case in our isometry circuits, where they are simple bit‑flips in the tableau of the classical algorithm.
The two unrestricted PUI-circuits used in the isometry of Figure~\ref{fig:isometry} are plotted in Figure~\ref{fig:unrestricted_pui}.

A formal description of the classical algorithm for creating the unrestricted partial unary iteration circuit implementing $U$ and determining the map $g$ is given in Algorithm~\ref{alg:partial_unary_iteration}.
It works similarly to the restricted version described above by recursively traversing parts of the same binary tree.
However, we change the branching logic at a node depending on the overlaps of the target interval with the two half-intervals.

If both the left and right half-intervals intersect with the target interval, we use the same branching primitive as explained above.
This is the case in the black dotted box on the right of Figure~\ref{fig:unrestricted_pui} where the left and right half-intervals are the one element sets consisting of $\ket{4}^a$ and $\ket{5}^a$, respectively, both lying within the target interval $[\ket{4}^a, \ket{6}^a]$.

In case only the right half-interval intersects with the target interval, we apply the box primitive, but only continue the recursion on the right half-interval.
In the example on the right of Figure~\ref{fig:unrestricted_pui} only the right half-interval $[\ket{4}^a, \ket{7}^a]$ defined by the address qubit $a_0$ being in the state $\ket{1}^{a_0}$ has overlap with the target interval $[\ket{4}^a, \ket{6}^a]$ and consequently, all operations are controlled on~$\ket{1}^{a_0}$.
This assures that the partial unary iteration acts trivially on address-register content that is smaller than the target interval.

While the two cases above are the same as for restricted PUIs, the approach differs if the left half-interval intersects with the target interval, but the right one does not.
In this case, we do not employ the box primitive and instead continue the recursion on the left half-interval, i.e., we ignore the address qubit that would go into the box.
In Figure~\ref{fig:unrestricted_pui} on the left, this situation appears for $a_0$ corresponding to $[\ket{0}^a, \ket{7}^a]$ whose right half-interval does not intersect the target interval $[\ket{0}^a, \ket{3}^a]$, and, a second time on the right, for $[\ket{6}^a, \ket{7}^a]$ where the right half-interval consisting of $\ket{7}^a$ lies outside of $[\ket{4}^a, \ket{6}^a]$.
The change in the circuit compared to the restricted PUI amounts to the removal of the purple controls on the address qubit from Figures~\ref{fig:restricted_pui}. 
Doing so results in the target unitary not only being applied for states where the address qubit is $\ket{0}$, but also for all states where this qubit is $\ket{1}$.
This means that we apply target unitaries at address states greater than the target interval.

Once a leaf, corresponding to an address-state $\ket{j}^a$, is reached, the elements we need to add to $L$ can be inferred from $j$ by going through all nodes in the path to the leaf where we ignored the address qubit and taking all combinations of $0$ and $1$ at those address positions such that at least one of them is $1$.
All those elements then get mapped to $j$ under $g$.

It remains to determine the resources required for implementing the restricted and unrestricted partial unary iteration circuits.
The number of ancilla qubits is at most $w-1$ in both cases.
When it comes to the Toffoli gate count, we do not count the gates required for a single partial unary iteration circuit, but concentrate on the situation in the isometry circuits instead.
We assume to successively iterate over intervals of size $m\coloneqq 2^l$ for some $l\leq w$, starting from $\ket{0}^a$.
We restrict to this case because it is more efficient to move in intervals whose size is a power of two than to use arbitrary interval lengths.
The latter requires traversing the same tree node twice in certain situations, leading to increased Toffoli cost.
As the circuits take a very simple form, the situation is also more easily analyzed than the case of general restricted and unrestricted partial unary iterations.

We start with the case where the successive partial unary iterations cover all possible address-register values, i.e., there are $2^{w-l}$ intervals.
Each interval may be understood as selecting some value on the first $w-l$ address qubits and then doing a full controlled unary iteration on the remaining $l$ qubits, the latter requiring $m-1$ Toffolis (resp. $m-2$ in the uncontrolled case) per interval.
In the case of the restricted partial unary iterations, selecting a value on the first $w-l$~qubits requires $w-l-1$~ANDs, each with a cost of a single Toffoli.
In total, one thus has a Toffoli cost for the restricted partial unary iteration of
\begin{equation*}
    \frac{S}{m}\left(m +\log(S/m)-2\right)
\end{equation*}
with $S\coloneqq 2^w$.

For the unrestricted case, notice that Toffolis are required for selecting a state on the first $w-l$~qubits if and only if there are at least two $\ket{1}$-qubits among them.
For such a state with $q$~qubits in state $\ket{1}$, the number of required Toffolis is $q-1$.
There are $2^{w-l}(w-l)/2$ occurrences of $\ket{1}$-qubits within the first $w-l$ qubits in all possible states, $w-l$ of which belong to states with exactly one $\ket{1}$.
Hence, there are a total of $2^{w-l-1}(w-l) - (w-l)$ relevant occurrences of~$\ket{1}$ within $2^{w-l} - (w-l) - 1$ states, yielding
\begin{equation}
\label{eq:toff}
    2^{w-l}\left(\frac{w-l}{2} - 1\right) + 1
\end{equation}
Toffolis for selecting on the first $w-l$ address qubits.
This results in a total cost for the unrestricted case of
\begin{equation*}
    \frac{S}{m}\left(m +\frac{\log(S/m)}{2}-2\right)
\end{equation*}
Toffolis.

So far, we assumed that the successive intervals cover all of the $S=2^w$ values of the address register.
However, in the isometry circuit of Section~\ref{sec:isometry_circuit}, it is only iterated over $s\leq S$ elements, which need not be a power of two.
In this case, assuming $2^{w-1}<s<2^w=S$, a simple inductive argument in Appendix~\ref{sec:proof} shows that the number of Toffolis required for selecting on the first $w-l$~qubits in the unrestricted partial unary iteration may be upper bounded by
\begin{equation}
\label{eq:bound}
    \biggl \lceil \frac{s}{m} \biggl \rceil \left(\frac{w-l}{2} - 1\right) + 1.
\end{equation}
The argument is based on the observation that selecting intervals closer to $s$ contributes more Toffolis than for intervals closer to zero.
For the full unrestricted partial unary iteration, one thus gets the bound
\begin{equation}
\label{eq:final_count_pui_unrestricted}
    \biggl \lceil \frac{s}{m} \biggl \rceil \left(m +\frac{\log(S/m)}{2}-2\right).
\end{equation}
Recalling that in the isometry setup $w=\lceil\log(s)\rceil$ and therefore $S=\tilde{s}$, equation \eqref{eq:algo} and the fact that the last individual PUI may be required to be restricted, leads to the bound \eqref{eq:toffoli_cost}.
The change in Toffoli cost in case the last PUI needs to be restricted is upper bounded by $\log(S/m)$.
For the restricted case, it is easy to see that the total PUI cost is bounded by
\begin{equation}
\label{eq:final_count_pui_restricted}
    \biggl \lceil \frac{s}{m} \biggl \rceil \left(m +\log(S/m)-2\right).
\end{equation}
We note that for small values of $s/m$ the bounds can be rather loose.
As a result, considering the bounds as functions of $m$, there is a value at which the upper bound is minimal, after which it becomes larger again.
This is in contrast to the actual Toffoli count, which is monotonically decreasing.

\section{Optimizing the full circuit}
\label{sec:further_optimizations}

So far, we have only looked at how to best implement the isometry step in the sparse state preparation scheme of Figure~\ref{fig:preparation_circuit}.
Combining equations~\eqref{eq:algo} and \eqref{eq:final_count_pui_unrestricted}, taking into account that the last PUI needs to be restricted which adds $\frac{1}{2}\log(\tilde{s}/m)$ Toffolis compared to the unrestricted version, the cost of the isometry is
\begin{multline*}
   \biggl \lceil \frac{s}{m} \biggl \rceil \left(m +\frac{\log(\tilde{s}/m)}{2}-2\right) + \frac{\log(\tilde{s}/m)}{2} + s - \biggl \lceil \frac{s}{m} \biggl \rceil + \log(\tilde{s}/2) \\
    \leq \bigg \lceil \frac{s}{m} \biggl \rceil \left(2m +\frac{\log(\tilde{s}/m)}{2}-3\right) + \log\left(\frac{\sqrt{\tilde{s}^3}}{2\sqrt{m}}\right),
\end{multline*}
where $m$ is the largest power of two less than or equal to the number of qubits left from the $n$-qubit main register after the dense state is prepared, i.e., $m \leq n - \lceil\log(s)\rceil$.
This is a sizeable reduction compared to the worst-case Toffoli cost of $s(\lceil{\log(s)}\rceil -1)$ in~\cite{Malvetti2021}.
We recall that $m$ is chosen as a power of two to ensure that in each PUI a full subtree is explored.
Otherwise some nodes of the subtree could be visited twice in consecutive intervals, and possibly require Toffolis in both cases.

For implementing the isometry, $\lceil{\log(s)}\rceil - 1$ ancilla qubits are needed.
These ancillas can obviously also be used for the dense state preparation, which means that $n + \lceil{\log(s)}\rceil - 1$ qubits are available for this step without increasing the total qubit count of the overall sparse state preparation routine.
More globally, state preparation is usually a subroutine at the beginning of a larger quantum circuit, and the number of available qubits is determined by the later parts of the circuit.
In this section, we will therefore consider how to best implement the combined sparse state preparation steps with varying numbers of ancilla qubits.

\begin{figure}[t]
    \centering
    \begin{adjustbox}{width=\textwidth}
    \begin{tikzpicture}
        \begin{yquant}[/yquant/register/separation=1.5mm]
            qubit {} state[3];
            qubit {$\ket{0}^r$} b;
            qubit {$\ket{\psi}$} psi;
            nobit anc;

            init {$\ket{0}$} (state[0-2]);
            ["north:$b$" {font=\protect\footnotesize, inner sep=1pt}]
            slash b;
            ["north:$b$" {font=\protect\footnotesize, inner sep=1pt}]
            slash psi;
            box {$\phi_0$} b;
            [name = i0, operator style={only at={1,2}{rounded corners}}] box {\Ifcase\idx\space$R_y(\phi_0)$\Or$\phi_0$\Else$\psi$\Fi} state[0], b, psi;
            box {M$_X$} b;
            [name = i1l, operator style={only at={0}{rounded corners}}] box {\Ifnum\idx<1{\rotatebox{90}{QRO(A)M$_1$}}\Else$\phi_1^{0,1}$\Fi} state[0], b;
            [name = i1, operator style={only at={1,2}{rounded corners}}] box {\Ifcase\idx\space$R_y(\phi_1^{0,1})$\Or$\phi_1^{0,1}$\Else$\psi$\Fi} state[1], b, psi;
            box {M$_X$} b;
            [name = i2l, operator style={only at={0}{rounded corners}}] box {\Ifnum\idx<1{\rotatebox{90}{QRO(A)M$_2$}}\Else$\phi_2^{0,\dots,3}$\Fi} (state[0-1]), b;
            [name = i2, operator style={only at={1,2}{rounded corners}}] box {\Ifcase\idx\space$R_y(\phi_2^{0,\dots,3})$\Or$\phi_2^{0,\dots,3}$\Else$\psi$\Fi} state[2], b, psi;
            box {M$_X$} b;
            [this subcircuit box style={draw, dashed}]
            subcircuit {
                qubit {} state[3];
                qubit {} b;
                qubit {} psi;
                nobit anc;
                align -;
                init {$\ket{1}$} anc;
                [name = i4l, operator style={only at={0}{rounded corners}}] box {\Ifnum\idx<1{\rotatebox{90}{QRO(A)M$_3$}}\Else$\phi_3^{0,\dots,7}$\Fi} (state[0-2]), b;
                [name = i4, operator style={only at={0,1}{rounded corners}}] box {\Ifnum\idx<1$\phi_3^{0,\dots,7}$\Else{\Ifnum\idx<2$\psi$\Else$R_z(\phi_3^{0,\dots,7})$\Fi}\Fi} b, psi, anc;
                box {M$_X$} b;
                text {$\ket{1}$} anc;
                discard anc;
                \draw (i4-0) -- (i4-1) -- (i4-2);
                \draw (i4l-1) -- (i4l-0);
            } (-);
            box {\rotatebox{90}{SignFixes}} (state); 
            output {$\ket{\psi}$} psi;
            output {$\ket{0}^r$} b;
            output {$\ket{\theta}$} (state);
        \end{yquant}
        \draw (i0-0) -- (i0-1) -- (i0-2);
        \draw (i1-0) -- (i1-1) -- (i1-2);
        \draw (i2-0) -- (i2-1) -- (i2-2);
        \draw (i1l-1) -- (i1l-0);
        \draw (i2l-1) -- (i2l-0);
    \end{tikzpicture}
    \end{adjustbox}
    \caption{Grover–Rudolph type dense-state preparation circuit on three system qubits, inspired by~\cite{berry2024, Low2024}. 
    The rotation angles $\phi_0, \phi_1^{0,1}, \dots, \phi_3^{0,\dots,7}$ are loaded into an angle register that has the same size $b$ as the register that holds the phase-gradient state $\ket{\psi}$. 
    The states in those two registers are then used as resource states for applying $R_y$ and $R_z$ rotations via the phase-gradient method~\cite{Low2024, Sanders2020}. 
    The first angle $\phi_0$ is encoded using X gates. 
    To load the angles that depend on the contents of the previously rotated main‑register qubits, QRO(A)Ms are used. 
    After a rotation, we measure the angle-register in the $X$-basis (M$_X$), which yields a clean $\ket{0}^r$ and introduces $\pm1$ phases on the basis states of the main register. 
    These phases can be inferred from the measurement results. 
    While the first three rotations are required to set the correct amplitudes of the target state, the last rotation in the dashed box prepares the correct phases. 
    At the end, we fix the signs. 
    As before, we round the corners of gates whose qubits are not changed by the gate.}
    \label{fig:dense_prep}
\end{figure}
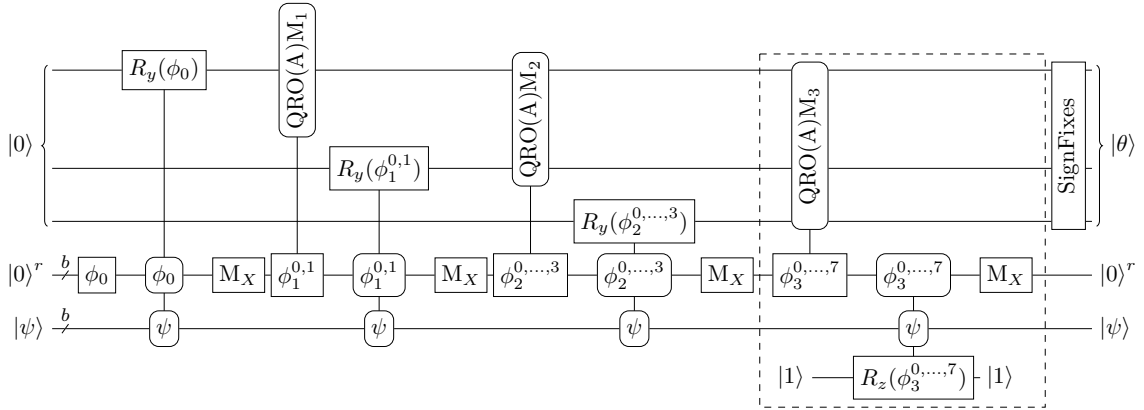

We start by briefly explaining the Grover-Rudolph strategy for dense state preparation, as depicted in Figure~\ref{fig:dense_prep} in its recent form~\cite{berry2024, Low2024}, where the rotations are implemented with the help of a phase-gradient state $\ket{\psi}$.
Given the dense state $\ket{\theta} = \sum_{i=0}^{2^l - 1} c_i \ket{f(i)}$ on $l\coloneqq \lceil\log(s)\rceil$ qubits, a two-step process is applied, the first of which prepares the correct amplitudes $\ket{\theta'} = \sum_{i=0}^{2^l - 1} |c_i| \ket{f(i)}$ by a series of $l-1$ sequential $R_y$-rotations, followed by the second step (dotted box in Figure~\ref{fig:dense_prep}) that uses an $R_z$-rotation to encode the correct phases of $\ket{\theta}$.
Here, the $i$-th phase-gradient rotation together with the preceding QRO(A)M$_i$~\cite{Low2024, Babbush2018, Berry2019} may be seen as a family of controlled rotations with angles $(\phi_i^j)_{j\in [2^i]}$ where the rotation angle $\phi_i^j$ gets applied if the state in the register consisting of qubits $0,\dots, i-1$ is $\ket{j}$.

The rotations are done via the well-known phase-gradient method~\cite{Low2024, Sanders2020}, taking the $b$-qubit phase-gradient state $\ket{\psi}$ and a rotation angle $\phi$ stored in a $b$-qubit angle-register $\ket{\phi}^r$ to implement the mapping $\ket{\psi}\bra{\psi}\otimes \ket{\phi}\bra{\phi}^r \otimes R(\phi)$ via controlled addition/subtraction of the states in the angle- and phase-gradient registers.
For the $i$-th rotation, in order to apply different rotation angles $(\phi_i^j)_{j\in [2^i]}$ depending on the state $\ket{j}$ of the qubits $0,\dots, i-1$, one uses some form of QRO(A)M $\sum_{j\in [2^{i}]} \ket{\phi_i^j}\bra{0}^r \otimes \ket{j}\bra{j}$ to load $\ket{\phi_i^j}^r$ into the angle-register if the address state is $\ket{j}$.
By linearity, the following phase-gradient rotation then causes the desired rotation $\sum_{j\in [2^{i}]} \ket{j}\bra{j} \otimes R(\phi_i^j)$.
The register size $b$ is logarithmic in the accuracy of the stored rotation angles and, in most practical examples, is on the order of $20$.
The phase-gradient rotation itself requires $b$ Toffoli gates, which is negligible compared to the dominating QRO(A)Ms and is thus neglected in this work.
It suffices to prepare the phase-gradient state $\ket{\psi}$ once, as it can be reused.
The cost for preparing $\ket{\psi}$ is therefore ignored in all our resource estimates, and we also exclude the phase-gradient register from our qubit counts.
We refer to~\cite{Low2024, Sanders2020} for the definition of the phase-gradient state $\ket{\psi}$ and for how to compute the rotation angles.

Instead of completely uncomputing the angle-register after each rotation as done in~\cite{Low2024} we follow~\cite{berry2024} and measure it in the X-basis.
After this measurement $M_X$ the angle register can be reset to $\ket{0}^r$.
The measurement induces $\pm 1$ phases on the main-register basis states which may be inferred from the measurement results and the angle-register content prior to the measurement.
If QROAM is used for angle-loading, the junk qubits of the QROAM are treated analogously.
After classically tracing the signs through the circuit, at the end we can apply sign fixes to the $l$-qubit target basis-states $\ket{0}, \dots, \ket{2^l - 1}$.
In~\cite{Berry2019}, it is shown how to do these sign corrections via a QROM-lookup table on $l-1$ qubits or, if additional clean ancillas are available, via a construction similar to QROAM.
The latter has Toffoli cost $2^{l-r} + 2^r$ at the cost of $2^r + (l-r)$ clean ancillas for any $r<l$.
We will see that a variation of our isometry algorithm, where, among other small changes, the unrestricted partial unary iterations are replaced by restricted ones, instead allows these sign-fixes to be done within the isometry at no extra cost besides the slightly more expensive restricted partial unary iterations.

Before explaining this in more detail, we review methods for loading the rotation angles with different Toffoli–qubit trade-offs, the simplest of which is the original QROM proposal of~\cite{Babbush2018}.
For an address register of size $k$, it requires $2^k-2$ Toffolis and $k-1$ clean ancillas, which, for the circuit type in Figure~\ref{fig:dense_prep}, incurs a total cost of
\begin{equation*}
    \sum_{k=2}^{l} (2^k - 2) = 2^{l+1} - 2l - 2 < 2\tilde{s}
\end{equation*}
Toffolis with $l-1$ ancilla qubits.
If the final sign-fix is done with a lookup table of size $l-1$, another $2^{l-1}-2$ Toffolis are needed, leading to a total count bounded by $2.5\tilde{s}$.

If one has more ancilla qubits at hand, the clean QROAM version of~\cite{Berry2019} may be used to bring down the Toffoli cost.
Combining QROM with controlled swaps, it allows loading $b$ bits of data from a $k$-qubit address register using $2^{k-r} + b(2^r-1)$ Toffolis at the cost of $b(2^r-1) + (k-r)$ ancillas for some $r<k$.
The Toffoli count is minimized for $2^r \approx \sqrt{2^k/b}$, with the cost being roughly $2\sqrt{b2^k}$, which means that with enough ancillas one can obtain almost a quadratic improvement in the number of Toffolis over the original QROM.
Assuming for simplicity that we chose the same $r$ for all angle-lookups, the overall Toffoli cost sums up to
\begin{equation*}
    \sum_{k=2}^{l} (2^{k-r} + b(2^r-1)) = b(l-1)(2^r-1) + (2^l-2) 2^{1-r} < 2\tilde{s}/2^{r} + b(\lceil\log(s)\rceil -1) (2^r -1)
\end{equation*}
respectively $3\tilde{s}/2^{r} + (b\lceil\log(s)\rceil-b +1)(2^r-1)$ if the sign-fix at the end is done with the scheme of~\cite{Berry2019}.
The ancilla cost in both cases is $b(2^r-1) + (\lceil\log(s)\rceil - r)$.
Note that if, at each $k$, the value $2^r \approx \sqrt{2^k/b}$ is chosen optimally, one ends up with a leading-factor Toffoli cost of $2(1+\sqrt{2})\sqrt{2b\tilde{s}}$ without the the sign-fix.
Hence, it is possible to improve upon the dense state preparation with the original QROM variant almost quadratically.
Finally, note that $b(2^r-1)$ ancillas can be dirty~\cite{Low2024} if one can tolerate worse prefactors in the Toffoli cost.

There is another variant~\cite{berry2024} for loading the rotation angles for dense state preparation with an even better Toffoli cost using multiple angle-registers and controlled swaps.
Since the minimal number of ancilla qubits is not as controllable as in the clean QROAM approach above, we, however, decide not to analyse it further here.

Comparing the Toffoli costs of the isometry and dense-state-preparation steps, in case QROM is used for the latter, the costs are distributed relatively evenly, with the upper bound~\eqref{eq:toffoli_cost} for the isometry being roughly $2s$ and the dense-state-preparation requiring $2.5\,\tilde{s}$ Toffolis.
As seen in Figure~\ref{fig:isometry_improvement}, however, in practical examples the isometry cost is often close to $s$, making the QROM-based dense preparation step the dominant one.
As noted in Section~\ref{sec:isometry_circuit} this factor of two difference between the worst case and practical examples is due to the fact that usually there are enough suitable states for building a batch such that most of the time no Toffoli is required to create such a state.
The distribution of the Toffoli cost among the different parts of the sparse-state-preparation stack for this QROM-based variant can be seen on the leftmost bar in Figure~\ref{fig:numerics_whole_prep}.
Replacing QROM by QROAM with maximal $r$-value $r_{max}$ changes the picture, and as $r_{max}$ increases, the dense-state-preparation cost becomes increasingly negligible compared to the isometry cost.
Up to a certain $r_{max}$, which is mostly determined by $n$, the extra qubits required for QROAM do not significantly exceed the number of qubits needed for the isometry, and hence the Toffoli cost can be reduced without increasing the total qubit count.
As can be seen further to the right, beyond a certain point, trading additional ancilla qubits for fewer Toffolis in the dense-state-preparation gadget does not significantly improve the overall cost of sparse-state preparation.

Even though $n$ is often large enough that QROAM with sufficiently large $r$ values can be applied without excessive qubit overhead, it is sensible, especially when few qubits are available, to try to reduce the leading factors of the dense-state-preparation cost by other means.
We have seen that applying the sign fixes for uncomputation increases the leading prefactor of the Toffoli cost from $2$ to $3$ (or $2.5$ if standard QROM is used).
In Figure~\ref{fig:numerics_whole_prep}, this part of the cost is drawn in purple.
As described above, the primitive responsible for the cost is the sign-fix applied to the basis states according to the results of the $X$-basis measurements.
Since, in the worst-case cost analysis of our isometry circuit, we assume iterating over all basis states, it is natural to consider performing the sign fix-up on the fly within the PUIs during the isometry circuit by simply applying a Z~gate to a $\ket{1}$-ancilla for basis states that require a phase fix.
In order to do so, we have to slightly change our isometry circuit and the algorithm that builds it, because Algorithm~\ref{alg:isometry_finding} does not guarantee iterating over all basis states.
Indeed, basis states might start in or be mapped into the subspace without being iterated over.
To rectify this problem, we use Algorithm~\ref{alg:isometry_finding_with_phase}.
\begin{figure}[t]
    \centering
    \includegraphics{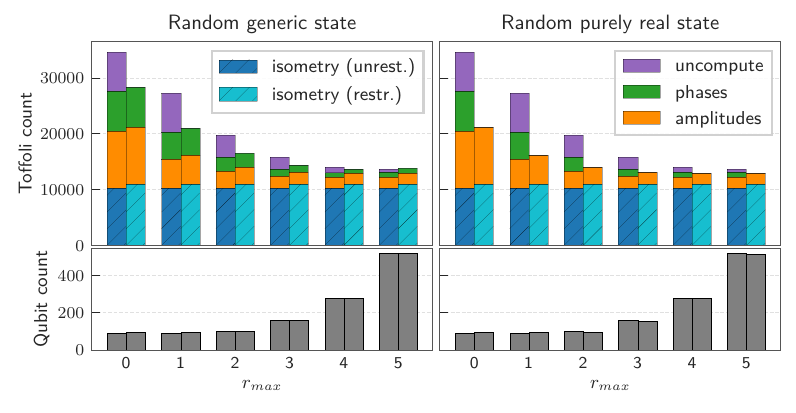}
    \caption{Toffoli and qubit counts of end-to-end sparse state preparation for a random complex (left) and purely real (right) $s$-sparse state on $n=80$~qubits for $s=10\,000$. 
    We distinguish the individual contributions of the dense state preparation and consider different values of $r_{max}$, the maximum $r$ available for the QROAMs. The case $r_{max}=0$ corresponds to QROM. 
    For the isometry both Algorithm~\ref{alg:isometry_finding} (restr.) and Algorithm~\ref{alg:isometry_finding_with_phase} (unrest.) are applied. 
    The resource count were calculated using Qualtran~\cite{harrigan2024}.}
    \label{fig:numerics_whole_prep}
\end{figure}

At the beginning of the algorithm, an extra ancilla qubit is set to $\ket{1}$.
Moreover, the unrestricted partial unary iteration is replaced by the restricted version, and the extra qubit of a basis state is set to $\ket{0}$ when a partial unary iteration iterates over the subspace-register content of that basis state.
If elements are in the subspace without having been part of a batch, we use the extra qubit to map them out of the batch and continue with the algorithm.
To see that this guarantees that all basis states are iterated over, note that if there are multiple states with the same subspace-register content and one of these states is part of a partial unary iteration that sets the extra qubit from~$\ket{1}$ to $\ket{0}$, then the same occurs for the others.
Each of those states has some qubit outside the subspace register and the extra qubit set to $\ket{1}$, and we must ensure that this remains true until the state is iterated over in a PUI where the sign is fixed.
One can easily check that this is indeed the case by noting that the only operation that can zero the non-subspace-register qubits of a state is the PUI.
One can either use an extra ancilla qubit set to $\ket{1}$ for applying the Z~gate, or, by keeping track of the sign values throughout the algorithm, use one of the conditionally clean qubits~\cite{Khattar2025} that can be temporarily transformed to $\ket{1}$.
The effect on the Toffoli count of outsourcing the uncomputation to the isometry for a random state can be seen in Figure~\ref{fig:numerics_whole_prep} (left).

If the state to prepare, $\ket{\Theta}$, is real, the phases to encode in the last step of the dense-state-preparation circuit (dotted box in Figure~\ref{fig:dense_prep}, green boxes in Figure~\ref{fig:numerics_whole_prep}) are $\pm 1$ signs.
It is therefore possible to remove this step and, analogously to the sign fix-up, encode the phase information in the isometry.
Since the unrestricted partial unary iteration is replaced by a restricted one, the bound for the Toffoli cost of the adapted isometry circuit is slightly increased to
\begin{equation*}
\biggl \lceil \frac{s}{m} \biggl \rceil(2m +\log(\tilde{s}/m)-3)
\end{equation*}
as can be seen from equations \eqref{eq:algo} and \eqref{eq:final_count_pui_restricted}.
Since all states are iterated over, the last term in \eqref{eq:algo} can be ignored as there is no need for the fix-up operation explained in Section~\ref{sec:isometry_circuit}.
This estimate, together with all other estimates from the paper and some from the literature, may be found in Table~\ref{tab:costing}.
\begin{figure}[t]
    \centering
    \begin{subfigure}[c]{0.55\textwidth}
        \centering
        \includegraphics{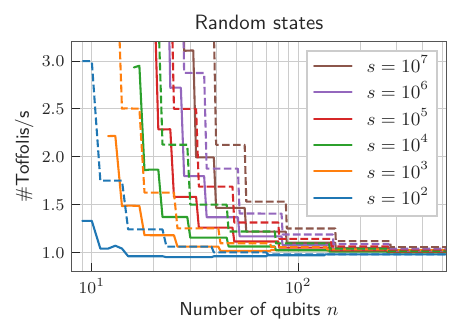}
    \end{subfigure}
    \hfill
    \begin{subfigure}[c]{0.40\textwidth}
        \centering
        {\renewcommand{\arraystretch}{1.15}%
        \begin{tabular}{|c||c|c|}
        \hline
        \multirow{2}{*}{$s$} &
        \multicolumn{2}{c|}{Time (seconds)} \\ 
        \cline{2-3}
        & unrestr. & rest.\\ 
        \hline
        \hline
        $10^2$ & 0.0004 & 0.0003 \\ 
        $10^3$ & 0.0043 & 0.0037 \\ 
        $10^4$ & 0.0897 & 0.0981 \\ 
        $10^5$ & 4.5485 & 4.7399 \\ 
        $10^6$ & 26.056 & 22.356 \\
        $10^7$ & 1488.7 & 813.22 \\ 
        \hline
        \end{tabular}}
        \vspace{0.3cm}
    \end{subfigure}
    \caption{(Left) Toffoli-cost comparison of Algorithm~\ref{alg:isometry_finding} (solid) and Algorithm~\ref{alg:isometry_finding_with_phase} (dashed) on random $s$-sparse states with $n$ qubits, divided by $s$. 
    (Right) Wallclock-times for running implementations of Algorithm~\ref{alg:isometry_finding} (unrestr.) and Algorithm~\ref{alg:isometry_finding_with_phase} (rest.) on $s$~random basis states averaged over different values of $n$. 
    The times were measured on a CPU-node with two AMD EPYC 7452. 
    The implementations use a single thread for $s\leq 10^5$ and multi-threading for larger $s$. 
    The multi-threaded version of the restricted algorithm is faster than the unrestricted version due to the way multi-threading is implemented.}
    \label{fig:numerics_isometry}
\end{figure}

In practice, as can be seen for random states on the left in Figure~\ref{fig:numerics_isometry}, the Toffoli cost of the restricted version of the isometry goes to $s$ for sufficiently large $n$, analogous to the unrestricted version.
However, for a given $s$ there is a regime of $n$ where the restricted isometry can require substantially more Toffoli gates than the unrestricted version.
For a given state to prepare, it is therefore sensible to consider both options, taking into account the number of available qubits, to find the most suitable sparse-state-preparation circuit.
In Figure~\ref{fig:numerics_whole_prep}, this is done for both a random generic (complex) and a purely real state.

\section{Conclusion \& Outlook}
\label{sec:conclusion}

In this work, we have developed algorithms for sparse state preparation with improved Toffoli cost.
We achieve this by refining an algorithm of Malvetti et al.~\cite{Malvetti2021} for finding and implementing isometries from the full space into a subspace whose qubit number is logarithmic in the number of states to prepare, replacing groups of individual multi-controlled-X gates with more efficient primitives we call partial unary iterations.
The resulting quantum circuits are, to the best of our knowledge, the most efficient sparse state preparation circuits in terms of Toffoli and qubit counts in the literature.

We introduce two simple modifications of the well-known unary iteration circuit~\cite{Babbush2018}, which we call restricted and unrestricted partial unary iterations, that might be of independent interest.
They iterate over intervals of values of an address register and apply unitaries on a target register controlled on the interval values.
In the unrestricted version, an improved Toffoli cost is achieved by allowing the circuit to change the target register content for address register values greater than the biggest value in the interval.
This is acceptable for states where one knows that no components with such address register content are present, or, as in our case, the address register values are iterated over from left to right in successive intervals, allowing prior unwanted unitaries applied to the target register to be undone in a later interval.

Finally, we review methods for implementing the dense state preparation step that constitutes the second part of the sparse state preparation algorithm and present some optimizations for the joint cost of the overall algorithm, particularly in the case of real target states.
This is achieved by outsourcing the uncomputation of the angle-register and, for real target states, the encoding of the phase information to the isometry circuit, requiring some changes in the isometry.

Following the standard approach for resource estimates for fault-tolerant algorithms, the discussion in this paper is focused on the number of Toffoli gates as the figure of merit.
As recent advances in generating magic states using cultivation~\cite{gidney2024, sahay2025} and advanced distillation techniques~\cite{ruiz2025} considerably lower the time and space footprint of magic-state creation, the number of Toffoli gates within an algorithm might lose its status as the main proxy for the algorithm's runtime~\cite{huggins2025}.
Therefore, in the future, adapted metrics for measuring the expected runtime of algorithms~\cite{mcardle2025, huggins22025} will be needed, and the cost of existing algorithms reevaluated.

Staying within the Toffoli‑and‑qubit‑count approach to resource estimation, the question of whether sparse state preparation can achieve Toffoli cost sublinear in $s$ in the practically relevant regime of limited ancilla‑qubit numbers remains open.

The developed state‑preparation circuits, together with the required sub‑circuits, are available as Qualtran‑Bloqs~\cite{rupprecht2025}, with the computationally more expensive classical algorithms for finding the isometry outsourced to SIMD‑accelerated, multi‑threaded Rust code.
While the runtime of these implementations allows us to find isometry circuits for sparse states with $s=10^7$~non‑zero basis states in a few minutes on reasonable hardware, we believe that there is room for improvement, potentially via efficient GPU code.
Achieving such speed-ups is necessary for running the algorithms on states with $s>10^7$ in an affordable amount of time.

In conclusion, we expect our sparse state preparation scheme to be useful as an efficient primitive in fault-tolerant quantum algorithms, such as quantum simulation and linear system solvers.
\\

\textbf{Note:} After this work appeared as a preprint, the question of
whether sparse state preparation with Toffoli cost sublinear in $s$
is possible was answered affirmatively by~\cite{li2026, luo2026}, which provide algorithms with Toffoli count scaling as $O(\sqrt{s})$, exploiting Toffoli/ancilla
trade-offs.

\section*{Acknowledgements}

This work was funded by the Ministry of Economic Affairs, Labour and Tourism Baden Württemberg through the Competence Center Quantum Computing Baden-Württemberg (KQCBW).
We thank Benjamin Desef for valuable discussions and comments on the manuscript.
The circuit diagrams in the paper were created with the \texttt{yquant} package~\cite{desef2021}.

\section*{Author contributions}
F.R. developed the idea, implemented it, and wrote the manuscript. S.W. supervised the work and provided feedback on the manuscript. Open-weight LLMs were used for language checking and for minor assistance with coding.

\bibliography{output}

\newpage

\appendix

\section{AND gates}
\label{sec:logical_and}

In Figure \ref{fig:log_and} we recall the definition of the AND gate primitive and its measurement-based adjoint as introduced in~\cite{Babbush2018, Gidney2018}.
\begin{figure}[h]
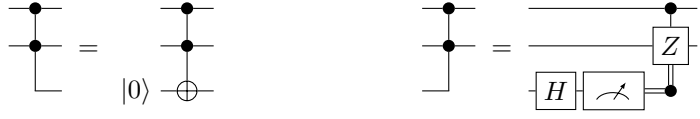

    \centering
    \begin{yquantgroup}[every control/.append style={radius=0.8mm}]
        \registers{
            qubit {} a[2];
            qubit {} b;
        }
        \circuit{
            discard b;
            init b | a;
        }
        \equals
        \circuit{
            init {$\ket0$} b;
            cnot b | a;
        }
        \equals[$\;\;\;\;\;\;\;\;\;\;\;\;\;\;\;\;\;\;\;\;\;\;\;\;$]
        \circuit{
            discard b | a;
        }
        \equals
        \circuit{
            h b;
            measure b;
            z a[1] | a[0], b;
            discard b;
        }
    \end{yquantgroup}
    \caption{AND gate and measurement-based inverse.}
    \label{fig:log_and}
\end{figure}

\section{Proof of bound}
\label{sec:proof}

In this section, we give the inductive proof of the bound \eqref{eq:bound}.
For this, we let $h(x)$ denote the Hamming weight of the binary representation of a natural number $x$.
The bound follows from the following Lemma with $k\coloneqq w-l$ and $r\coloneqq \lceil \frac{s}{m}\rceil$, for which we have 
\begin{equation*}
    2^{k-1}=\frac{2^{w-1}}{2^l} = \frac{2^{w-1}}{m}\leq \biggl \lceil\frac{s}{m}\biggr\rceil \leq \frac{2^w}{m} = 2^k.
\end{equation*}
\begin{lemma}
    Let $k\in\mathbb{N}$.  If $r$ is a natural number with $2^{k-1}\le r\le 2^{k}$ and $C_r\coloneqq\{x\in\mathbb{N}_{< r}\mid h(x)>1\}$, then $S_r\coloneqq\sum_{x\in C_r}(h(x)-1)\le r\left(\frac{k}{2}-1\right)+1$.
\end{lemma}
\begin{proof}
    We prove the statement by induction on $k$.
    For $k=1$, the statement is easy to verify.
    For the induction step, define $\tilde{C}_r\coloneqq\{x\in C_r\mid x\ge 2^{k-1}\}$ and notice that
    $S_r = S_{2^{k-1}} + \sum_{x\in \tilde{C}_r} (h(x) -1)$.
    Let $x\in\tilde{C}_r$ and define $y_x\coloneqq x-2^{k-1}$.
    Then $h(x)=h(y_x)+1$, and using the induction hypothesis we have
    \begin{align*}
        \sum_{x\in \tilde{C}_r} (h(x) - 1) = (r-2^{k-1} - 1) + \sum_{x\in \tilde{C}_r} (h(y_x) -1) &= (r-2^{k-1} - 1) + \sum_{y\in C_{r-2^{k-1}}} (h(y) -1)\\
        &\leq (r-2^{k-1} - 1) + (r-2^{k-1})\left(\frac{k-1}{2} - 1\right)+1.
    \end{align*}
    Here we use the fact that $|\tilde{C}_r|=r-2^{k-1}-1$, that the elements $y_x$ which do not belong to $C_{r-2^{k-1}}$ are precisely those with $h(y_x)=1$, and the induction hypothesis.
    Putting everything together, using $2^{k}\geq r$ and, once more, the induction hypothesis yields
    \begin{align*}
        S_r &= S_{2^{k-1}} + \sum_{x\in \tilde{C}_r} (h(x) -1) \\
        &\leq 2^{k-1}\left(\frac{k-1}{2} - 1\right)+1 + (r-2^{k-1} - 1) + (r-2^{k-1})\left(\frac{k-1}{2} - 1\right)+1 \\
        &= r\left(\frac{k}{2} - \frac{1}{2}\right) - \frac{2^{k}}{2} + 1 \leq r\left(\frac{k}{2} - 1\right) + 1.
    \end{align*}
\end{proof}

\section{Algorithms}

\DecMargin{0.8em}

\begin{algorithm*}[H]
    \DontPrintSemicolon
    \SetAlgoLined
    \KwData{Tableau of $n$-bit bitstrings $[\ket{C_i}]_{i\in [s]}$ written as $\ket{C_i}\coloneqq \ket{C_i'}^s\ket{C_i''}^r$ with $\ket{C_i'}^s$ having $l\coloneqq \lceil{\log(s)}\rceil$ bits.}
    \KwResult{Quantum circuit $G$ and bijection $f\colon [s] \to S$ where $S\subseteq [2^{\lceil{\log(s)}\rceil}]$ such that $G\ket{C_i} = \ket{f(i)}^s\ket{0}^r$ for all $i\in [s]$.}
    \BlankLine
    \vspace{-0.4mm}
    $G \leftarrow [\,]$.\;
    $k \leftarrow 0$.\;
    $batch \leftarrow [\,]$.    \tcp*[r]{Row indices of batch elements in tableau.}
    $m \leftarrow \max_{p\in \mathbb{N}}(2^p \mid 2^p \leq n - l)$. \tcp*[r]{Maximum batch size.}
    \While{True}{
        $b \leftarrow |batch|$.\tcp*[r]{Current batch size.}
        \If{$b = m$}{
            $\text{PUI}_{k-b}^{k-1}\leftarrow$ unrestricted PUI, mapping $\ket{C_j} = \ket{k_j}^s\ket{e_{q_j}}^r \mapsto \ket{k_j}^s\ket{0}^r$ for all $j\in batch$.\;
            $G$.append(PUI$_{k-l}^{k-1}$).\;
            $[\ket{C_i}]_{i\in [s]} \leftarrow [\text{PUI}_{k-b}^{k-1}\ket{C_i}]_{i\in [s]}$. \tcp*[r]{Update tableau.}
            $batch \leftarrow [\,]$.\;
        }
        \If{$\ket{C_i} = \ket{k_i}^s\ket{0}^r$ for all $i\in [s]\setminus batch$,}{ 
            $\widetilde{\text{PUI}}_{k-b}^{k-1}\leftarrow$ restricted PUI, mapping $\ket{C_j} = \ket{k_j}^s\ket{e_{q_j}}^r \mapsto \ket{k_j}^s\ket{0}^r$ for all $j\in batch$.\;
            $G$.append($\widetilde{\text{PUI}}_{k-l}^{k-1}$).  \tcp*[r]{Restricted PUI: Higher subspace states untouched.}
            $[\ket{C_i}]_{i\in [s]} \leftarrow [\widetilde{\text{PUI}}_{k-b}^{k-1}\ket{C_i}]_{i\in [s]}$.\;
            \If{$\ket{C_j} = \ket{k-1}^s\ket{e_{b-1}}^r$ for some $j\in [s]$,}{
                \tcp*[r]{Both states $\ket{k-1}^s\ket{0}^r$ and $\ket{k-1}^s\ket{e_{b-1}}^r$ present before PUI.}
                $\ket{k'}^s\leftarrow$ any bitstring such that $\ket{C_i} \neq \ket{k'}^s\ket{0}^r$ for all $i\in [s]$. \;
                $\text{CMX}_{l+b-1} \leftarrow$ multi-target CX with control $l+b-1$ mapping $\ket{C_j} \mapsto \ket{k'}^s\ket{e_{b-1}}^r$.\;
                $[\ket{C_i}]_{i\in [s]} \leftarrow [\text{CMX}_{l+b-1}\ket{C_i}]_{i\in [s]}$.\;
                MCX$^{l+b-1}\leftarrow$ multi-control CX with target $l+b-1$, mapping $\ket{C_j} \mapsto \ket{k'}^s\ket{0}^r$.\;
                $[\ket{C_i}]_{i\in [s]} \leftarrow [\text{MCX}^{l+b-1} \ket{C_i}]_{i\in [s]}$.\;
                $G$.extend([CMX$_{l+b-1}$, MCX$^{l+b-1}$]).\;
            }
            $f(i) \leftarrow k_i$ where $\ket{C_i}=\ket{k_i}^s\ket{0}^r$ for all $i \in [s]$.\;
            \Return{$G$, $f$}.\;
        }
        \If{there exists $j \in [s] \setminus batch$ such that $\ket{C_j''}^r \neq \ket{0}^r$,}{
            \If{$r \neq l + b$,}{
                $G$.append(SWAP$_r^{l+b}$). \tcp*[r]{Swap qubits $r$ and $l+b$.}
                $[\ket{C_i}]_{i\in [s]} \leftarrow [\text{SWAP}_r^{l+b}\ket{C_i}]_{i\in [s]}$.\;
            }
        }
        \Else{
            $j\leftarrow$ any $j \in [s] \setminus batch$ such that $\ket{C_j''}^r \neq \ket{0}^r$.\;
            $u \leftarrow$ any bit-index with $l\leq u<l+b$ at which $\ket{C_j}$ is non-zero.\;
            $v \leftarrow$ any bit-index at which $\ket{C_j}$ and $\ket{C_w}=\ket{k_w}^s\ket{e_{u-l}}^r$ from the batch differ.\;
            \tcp*[r]{$\ket{C_w}$ in batch by construction.}
            \If{$\ket{C_j}$ is zero at $v$}{
                $G$.append(Toff$_{u,\bar{v}}^{l+b}$). \tcp*[r]{Toffoli; negative control $v$.}
                $[\ket{C_i}]_{i\in [s]} \leftarrow [\text{Toff}_{u,\bar{v}}^{l+b}\ket{C_i}]_{i\in [s]}$.\;
            }
            \Else{
                $G$.append(Toff$_{u,v}^{l+b}$).\;
                $[\ket{C_i}]_{i\in [s]} \leftarrow [\text{Toff}_{u,v}^{l+b}\ket{C_i}]_{i\in [s]}$.\;
            }
        }
        $\text{CMX}_{l+b} \leftarrow$ multi-target CX with control $l+b$ mapping $\ket{C_j} \mapsto \ket{k}^s\ket{e_b}^r$.\;
        $G$.append(CMX$_{l+b}$).\;
        $[\ket{C_i}]_{i\in [s]} \leftarrow [\text{CMX}_{l+b}\ket{C_i}]_{i\in [s]}$.\;
        $batch$.append($j$).\;
        $k \leftarrow k + 1$.\;
    }
    \caption{Finding the isometry and bijection.}
    \label{alg:isometry_finding}
\end{algorithm*}

\begin{algorithm*}[H]
    \SetKwProg{Fn}{def}{\string:}{}
    \SetKwBlock{Match}{}{}
    \SetKwInput{Kw}{input}
    \SetKwFunction{PuI}{UPUI}
    \DontPrintSemicolon
    \SetAlgoLined
    \KwData{Interval $S\coloneqq [\ket{l}^a, \ket{r}^a]$ of address basis states from a $w$-qubit address register and corresponding quantum gates $U_l, \dots, U_r$ to be applied on a target register.}
    \KwResult{Quantum circuit $G$, subset $L\subseteq \{r+1, \dots, 2^w-1\}$ and mapping $g \colon L \to \{l, \dots, r\}$ such that $G = \sum_{i=l}^{r} \ket{i}\bra{i}^a\otimes U_i + \sum_{i\in L} \ket{i}\bra{i}^a\otimes U_{g(i)}$.}
    \BlankLine

    $G \leftarrow [\,]$.\;
    $I \leftarrow [\ket{0}^a, \ket{2^w-1}^a]$.\;

    \PuI{I, None, $\emptyset$}.\;
    \BlankLine
    \BlankLine

    \Fn{\PuI{interval, control, free}}{
        $d\leftarrow w - \log(|interval|)$.  \tcp*[r]{Current address qubit; big-endian.}
        \If{$interval = [\ket{i}^a, \ket{i}^a]$,}{
            \If{control is not None,}{
                \normalfont
                $G$.append(C$_{control}U_{i}$). \tcp*[r]{$U_i$ controlled on $control$.}
            }\Else{
                \normalfont
                $G$.append($U_{i}$).\;
            }
            \If{$|free| > 0$,}{
                \normalfont
                $J\leftarrow$ integers with $w$-bit binary representation equal to $i$ on indices not in $free$ and at least one $1$ at some position in $free$.\;
                $f(j) \leftarrow i$ for all $j\in J$.\;
            }
            \Return{}\;
        }
        $I_l, I_r \leftarrow$ \textnormal{left and right half-interval of} $interval$.\;
        \If{$I_l \cap S \neq \emptyset$ and $I_r \cap S \neq \emptyset$,}{
            \If{control is not None,}{
                \normalfont
                $G$.append(AND$_{control, \bar{d}}^{anc}$). \tcp*[r]{Negative control $d$; $anc$ is new ancilla.}
                \PuI{$I_l, anc, free$}.\;
                $G$.append(CX$_{control}^{anc}$).\;
                \PuI{$I_r, anc, free$}.\;
                $G$.append(${\text{AND}_{control,\bar{d}}^{anc}}^\dagger$).\;
            }\Else{
                \normalfont
                $G$.append(X$_d$). \tcp*[r]{X gate on $d$.}
                \PuI{$I_l, d, free$}.\;
                $G$.append(X$_d$).\;
                \PuI{$I_r, d, free$}.\;
            }
        }\ElseIf{$I_l \cap S \neq \emptyset$ and $I_r \cap S = \emptyset$,}{
            \PuI{$I_l, control, free \cup \{ d \}$}.\;
        }\Else(\tcp*[f]{Must have $I_l \cap S = \emptyset$ and $I_r \cap S \neq \emptyset$.}){
            \If{control is not None,}{
                \normalfont
                $G$.append(AND$_{control, d}^{anc}$).\tcp*[r]{$anc$ is new ancilla.}
                \PuI{$I_r, anc, free$}.\;
                $G$.append(${\text{AND}_{control, d}^{anc}}^\dagger$).\;
            }\Else{
                \PuI{$I_r, d, free$}.\;
            }
        }
    }
    \caption{Finding the circuit and mapping for unrestricted partial unary iteration.}
    \label{alg:partial_unary_iteration}
\end{algorithm*}

\begin{algorithm}[H]
    \DontPrintSemicolon
    \SetAlgoLined
    \KwData{Tableau of $n$-bit bitstrings $[\ket{C_i}]_{i\in [s]}$ written as $\ket{C_i}\coloneqq \ket{C_i'}^s\ket{C_i''}^r$ with $\ket{C_i'}^s$ having $l\coloneqq \lceil{\log(s)}\rceil$ bits.
    List of $\pm 1$-signs $[a_i]_{i\in [s]}$.}
    \KwResult{Quantum circuit $G$ and bijection $f\colon [s] \to [s]$ such that $G\ket{C_i} = a_i \ket{f(i)}^s\ket{0}^r$ for all $i\in [s]$.}
    \BlankLine
    $G \leftarrow [\,]$.\;
    $k \leftarrow 0$.\;
    $batch \leftarrow [\,]$.\;
    $m \leftarrow \max_{p\in \mathbb{N}}(2^p \mid 2^p \leq n - l)$\;
    $[\ket{D_i}^e]_{i\in [s]} \leftarrow [\ket{1}^e]_{i\in [s]}$.\tcp*[r]{Extend tableau with a $\ket{1}$-column as $(n+1)$-th qubit.}
    \While{True}{
        $b \leftarrow |batch|$.\;
        \If{$b=m$ or $k=s$,}{
            $\widetilde{\text{PUI}}_{k-b}^{k-1}\leftarrow$ restricted PUI, mapping $\ket{C_j}\ket{D_j}^e = \ket{k_j}^s\ket{e_{q_j}}^r\ket{D_j}^e \mapsto a_j\ket{k_j}^s\ket{0}^r\ket{0}^e$ for all $j\in batch$.\tcp*[r]{Fix sign; set extra qubit to $\ket{0}$.}
            $G$.append($\widetilde{\text{PUI}}_{k-b}^{k-1}$).\;
            $[(a_i,\ket{C_i}, \ket{D_i}^e)]_{i\in [s]} \leftarrow [\widetilde{\text{PUI}}_{k-b}^{k-1}(a_i,\ket{C_i},\ket{D_i}^e)]_{i\in [s]}$. \tcp*[r]{Update all tableaus.}
            $batch \leftarrow [\,]$.\;
            \If{$k=s$,}{
                $f(i) \leftarrow k_i$ where $\ket{C_i}=\ket{k_i}^s\ket{0}^r$ for all $i \in [s]$.\;
                \Return{$G$, $f$}.\;
            }
        }
        \If{there exists $j\in [s]$ such that $\ket{C_j}$ is non-zero at bit $r\geq l + b$,}{
            \If{$r \neq l + b$,}{
                $G$.append(SWAP$_r^{l+b}$).\;
                $[\ket{C_i}]_{i\in [s]} \leftarrow [\text{SWAP}_r^{l+b}\ket{C_i}]_{i\in [s]}$.\;
            }
        }
        \ElseIf{there exists $j \in [s] \setminus batch$ such that $\ket{C_j''}^r \neq \ket{0}^r$,} {
            $u \leftarrow$ any bit-index with $l\leq u<l+b$ at which $\ket{C_j}$ is non-zero.\;
            $v \leftarrow$ any bit-index at which $\ket{C_j}$ and $\ket{C_w}=\ket{k_w}^s\ket{e_{u-l}}^r$ in the batch differ.\;
            \If{$\ket{C_j}$ is zero at $v$}{
                $G$.append(Toff$_{u,\bar{v}}^{l+b}$).\;
                $[\ket{C_i}]_{i\in [s]} \leftarrow [\text{Toff}_{u,\bar{v}}^{l+b}\ket{C_i}]_{i\in [s]}$.\;
            }
            \Else{
                $G$.append(Toff$_{u,v}^{l+b}$).\;
                $[\ket{C_i}]_{i\in [s]} \leftarrow [\text{Toff}_{u,v}^{l+b}\ket{C_i}]_{i\in [s]}$.\;
            }
        }
        \Else{
            $j \leftarrow$ any $j\in [s] \setminus batch$ such that $\ket{D_j}^e=\ket{1}^e$. \tcp*[r]{Exists as $k<s$.}
            \For{$l \in batch \setminus \{j\}$ with $\ket{D_l}^e=\ket{1}^e$}{
                $G$.append(CX$_{q_l}^{l+b}$) where $\ket{C_l} = \ket{k_l}^s\ket{e_{q_l}}^r$. \tcp*[r]{Counter effect of later CX.}
            $[\ket{C_i}]_{i\in [s]} \leftarrow [\text{CX}_{q_l}^{l+b}\ket{C_i}]_{i\in [s]}$.
            }
            $G$.append(CX$_n^{l+b}$)  \tcp*[r]{CX gate with control $n$ and target $l+b$.}
            $[(\ket{C_i}, \ket{D_i}^e)]_{i\in [s]} \leftarrow [\text{CX}_n^{l+b}(\ket{C_i},\ket{D_i}^e)]_{i\in [s]}$.\;
        }
        $\text{CMX}_{l+b} \leftarrow$ multi-target CX with control $l+b$ mapping $\ket{C_j} \mapsto \ket{k}^s\ket{e_b}^r$.\;
        $G$.append(CMX$_{l+b}$).\;
        $[\ket{C_i}]_{i\in [s]} \leftarrow [\text{CMX}_{l+b}\ket{C_i}]_{i\in [s]}$.\;
        $batch$.append($j$).\;
        $k \leftarrow k + 1$.\;
    }
    \caption{Finding the isometry and bijection with phase fixing.}
    \label{alg:isometry_finding_with_phase}
\end{algorithm}

\newpage

\section{Resource costs for other sparse quantum state preparation algorithms}
\label{sec:other_papers}

Most prior papers in the literature proposing algorithms for sparse quantum state preparation did so using the asymptotic circuit size as their figure of merit.
As we are interested in constant factor Toffoli gate counts and ancilla numbers we briefly look at the corresponding resource estimates for the respective papers.
References to steps, algorithms or equations always refer to the corresponding paper under consideration.

We start with the group of papers~\cite{Mao2024, Gleining2021, Tubman2018, deVeras2022} in which the sparse state is built iteratively one basis state with its corresponding coefficient at a time.
As mentioned in the introduction, each of the $s$ basis states requires the implementation of a multi-controlled rotation gate or an equivalent construction which may always be reduced to a controlled rotation and a multi-controlled~X~gate~\cite{Gidney2018}.
The cost of the controlled rotation, as explained in Section~\ref{sec:introduction}, is given by $\lceil \log(s/\epsilon) \rceil$.
Here we look at the number of control qubits for the multi-controlled X gates.
Using the standard AND construction~\cite{Babbush2018, Gidney2018} a $w$-controlled X may be implemented with $w-1$ Toffolis and $w-2$ ancillas.
\begin{itemize}
    \item \cite{Mao2024}: The BE-QRAM algorithms works in batches of computational basis states of size $\lceil s/k \rceil$ with $k\coloneqq \lfloor \log(n) - \log(\log(n))\rfloor$. 
    The number of control qubits for each state (Step 3) is $2^k + 1$. Moreover, each batch has two $(n - 2^k)$-controlled X gates (Steps 2 and 4). 
    \item \cite{Gleining2021}: The control qubits are found by a halving procedure leading to at most $\lceil \log(d) \rceil$ controls.
    \item \cite{Tubman2018}: Declared cost of $O(n)$ without details. 
    Assume an $n$-controlled X gate.
    \item \cite{deVeras2022}: The controls in the CVO-QRAM algorithm are the qubits on which the current basis state is~$\ket{1}$. 
    Summed over all $s$ states the worst case control qubit number, i.e., the sum of the Hamming weights, is lower bounded by $ns/2$.

\end{itemize}

We now turn to the second group of algorithms~\cite{Fomichev2024, Ramacciotti2024, li2025, Malvetti2021} in which first a dense state is prepared which is then transformed into the sparse target state via an isometry.
Here we look at the Toffoli and qubit costs of the isometry.
In the paper~\cite{Fomichev2024} the figures of merit are already Toffoli and qubit gate counts.
Moreover, the approach in~\cite{Malvetti2021} was investigated in detail in Section~\ref{sec:isometry_circuit}.
\begin{itemize}
    \item \cite{Ramacciotti2024}: The isometry is a permutation of basis states which may be decomposed into cycles (Section IV). 
    The sum of the lengths of the cycles is upper bounded by $2s$, which is the case if there are $s$ transpositions, i.e., cycles of length two. 
    The non-Clifford cost for the implementation of a cycle of length $w$ is given by $w+1$ multi-controlled X gates with $n$~controls (Appendix B). 
    The approach uses an extra ancilla qubit. 
    In total the worst case Toffoli cost of the algorithm is $3s(n-1)$ with ancilla cost~$n-1$.
    \item \cite{li2025}: The permutation to implement can be chosen to consist of at most $s$ disjoint transpositions (Algorithm 1). 
    It is implemented in $m$ batches with $m$ being a power of two.
    Each transposition requires two $\lceil \log(2m)\rceil$-controlled X gates (Steps 2a and 2c in the proof of Lemma~12). 
    Moreover each batch uses a multi-controlled X gate 
    with $n-\lfloor\log(2m)\rfloor$ controls (Step 2b). 
    The total sum of Toffoli gates is therefore 
    \begin{equation*}
        \lceil s/m \rceil (n - \lfloor\log(2m)\rfloor - 1 + 2m(\lceil\log(2m)\rceil-1)).
    \end{equation*}
    Contrary to the paper where $m$ is set to $2^{\lfloor\log(\log(m)/4)\rfloor}$ for optimizing the circuit size, the Toffoli cost is optimal for $m\coloneqq s$, where there is only a single batch.
    The cost is then 
    \begin{equation*}
        2s(\lceil\log(2s)\rceil-1) +n  - \lfloor\log(2s)\rfloor-1
    \end{equation*}
    with ancilla count $\max(n-\lfloor\log(2s)\rfloor,\lceil\log(2s)\rceil) -1$.
    As we usually assume that $s \ll 2^n$, we use $n-\lfloor\log(2s)\rfloor$ in Table~\ref{tab:costing}.
    For brevity, we also drop the $-1$ in the Toffoli count there.
\end{itemize}

\end{document}